\documentclass[12pt, draftclsnofoot, onecolumn]{IEEEtran}
\usepackage[T1]{fontenc}

\usepackage{cite}
\usepackage{bookmark}
\usepackage{amsthm}
\newtheorem{theorem}{Theorem}
\newtheorem{lemma}{Lemma}
\newtheorem{corollary}{Corollary}
\newtheorem{remark}{Remark}
\theoremstyle{definition}
\newtheorem{definition}{Definition}

\ifCLASSINFOpdf
  	\usepackage[pdftex]{graphicx}
	\graphicspath{ {Figures/} }
	\usepackage[justification=centering]{caption}
\else
\fi
\usepackage{amsmath}
\usepackage{amssymb}
\usepackage{xspace}
\usepackage{bbm}
\usepackage{relsize}
\begin{document}
%
\title{PHP-Based Model and Analysis of
	Millimeter Wave Heterogeneous Cellular Network}
%
%
%

\author{Mehdi~Sattari
	and~Aliazam~Abbasfar,~\IEEEmembership{Senior~Member,~IEEE}
    	}

\maketitle

\begin{abstract}
In this paper, a novel Poisson hole process (PHP) modeling of wireless networks is proposed. We consider a more general hole shape, i.e. circular sector holes in a random direction, to capture the spatial separation between tiers in a millimeter wave (mmWave) heterogeneous cellular network (HCN). In this case, small cell BSs (SBSs) and macrocell BSs (MBSs) are distributed as a PHP and Poisson point process (PPP). Due to fundamental propagation differences between mmWave and microwave frequencies, we incorporate blockage and directional beamforming in the network model. Using tools from stochastic geometry, the evaluation of the coverage performance in a two-tier mmWave HCN is provided. Since the exact formulation of distance distribution, association probability and coverage probability based on PHP modeling of SBSs is not known, fairly accurate expressions for them are derived and validated by simulation results. Moreover, some interesting facts about the effect of holes on coverage probability and the relation between proposed hole configuration and prior models with circular holes is discovered. It turns out that the analysis based on the proposed PHP model can provide a more general study than prior works and the proposed approaches are more accurate for two-tier HCNs than independent PPP-based analysis.
\end{abstract}

\begin{IEEEkeywords}
Millimeter wave cellular network, heterogeneous cellular network, stochastic geometry, Poisson Hole Process, coverage probability.
\end{IEEEkeywords}

%
\IEEEpeerreviewmaketitle

\section{Introduction}
%
%
%
%
\IEEEPARstart{T}{he} explosive growth of mobile data traffic in recent years and spectrum shortage in microwave bands have motivated the use of new frequency bands for cellular communication. To that end, millimeter wave (mmWave) communications have been proposed to be an important part of the 5G mobile network to provide high data-rate communication services \cite{6824752},\cite{6515173},\cite{6732923}. Despite massive amounts of bandwidth available in mmWave bands, they have been  traditionally considered only for long distance point-to-point communication such as satellite or short range indoor communications. But recent channel measurements in \cite{6515173} have been revealed that they can be used for commercial cellular systems.

Due to fundamental propagation differences between mmWave and microwave frequencies, it is not possible to use prior models for cellular networks operating in microwave bands. Because of small wavelength and small antenna aperture, path loss in mmWave frequencies is larger than microwave frequencies. On the other hand, it is possible to use  antenna arrays in mmWave transceivers to provide high array gain, by beamforming, and compensate for large path loss.
Another distinguishing feature of mmWave frequencies is sensitivity to blockage. This effect needs to be included in models for analyzing mmWave networks. Channel measurements have revealed that there are noticeable differences in the line of sight (LOS) and non-line of sight (NLOS) links in these bands \cite{6515173}. For this reason, different path loss exponents in LOS and NLOS links are used for system models.
\subsection{Related Work}
Significant progress has been recently made to analyze the cellular networks performance using tools from stochastic geometry. Characterization of signal-to-interference-and-noise ratio (SINR) distribution in cellular networks was pioneered by \cite{6042301}, which provides a mathematical framework for downlink SINR distribution. Stochastic geometry-based analysis for mmWave cellular networks was first done in \cite{6489099}. Incorporating blockage effect and directional beamforming in network analysis of mmWave cellular system was applied in \cite{6932503} and a comprehensive mathematical model was provided. The results in \cite{6932503} revealed that mmWave cellular networks can outperform conventional cellular networks in terms of SINR coverage and average rate. Furthermore, a simplified blockage model named LOS ball model was introduced in \cite{6932503}. It was later extended by \cite{7110547}, in which a tractable model for data-rate was developed in self-backhauled mmWave cellular networks.

Heterogeneous cellular network (HCN) is modeled as a multi-tiered cellular network where the base stations (BSs) of each tier are randomly distributed and have a particular transmit power and spatial density. A general model for HCNs assumes K independent tiers of Poisson point process (PPP) distributed BSs, which is the most simplified model in the literature. Some tractable models for SINR distribution in a general K-tier downlink HCN have been developed in \cite{6171996}, \cite{5743604}, \cite{6287527}. Employing similar approaches from stochastic geometry used in microwave frequency and incorporating the distinguishing features of mmWave frequency, in \cite{7931577} heterogeneous downlink mmWave cellular networks was studied.

Although most of the stochastic geometry works on HCNs focus on PPP-based deployment of BSs, it is unrealistic to assume that the locations of the BSs corresponding to different tiers are completely uncorrelated. Therefore, the effect of spatial separation between tiers are usually overlooked and modeling HCNs accounting for the spatial dependence should be devised. In \cite{7110505} two different cases, inter-tier dependence and intra-tier dependence, are considered for a two-tier HCN. Each of the two cases introduces one form of dependence between the locations of macro and small cells. In such a setup, the small cell locations are modeled by a Poisson hole process (PHP) and Matern cluster process (MCP). A comprehensive analysis for PHP is developed in \cite{7557010}, where the Laplace transform of the interference is derived by considering the local neighborhood around a typical node of a PHP. In \cite{7922493}, lower bounds on the cumulative distribution function (CDF) of the contact distance in a PHP for two different cases are derived and compared with previous approximations.

PHP has also various applications in modeling wireless networks such as device-to-device (D2D) networks and cognitive radio networks. In D2D networks, D2D transmissions have an exclusion region around each transmitter and the remaining active D2D transmitters are modeled by a PHP \cite{7073589}. In cognitive radio networks, a secondary user equipment (UE) may transmit only when it is outside the primary exclusion regions \cite{6155562} to avoid harmful interference to primary UEs. In this case, the spatial distribution of active secondary UEs forms a PHP.

\subsection{Contribution}
In this paper, a novel PHP modeling of wireless networks and some useful insights of this model have been presented. Since mmWave signals are sensitive to blockage and BSs will be deployed in a more dense setting, prior PHP modeling of wireless networks with circular shaped holes may not be a tenable choice to capture spatial separation in mmWave HCNs. For instance, to fill the coverage holes and improve network capacity, an operator may deploy an SBS close to an MBS due to possible blockage. Such differences in mmWave signals motivated us to employ a more general and flexible hole shape in modeling PHP. To meet this demand, We have considered circular sector holes centered at MBS locations and in a random direction. It should be noted that our analysis based on this hole shape can be easily extended to the circular holes, which is the default configuration of holes in the definition of PHP \cite{Haenggi:2012:SGW:2480878}.

To provide a pertinent framework to evaluate the SINR distribution in a mmWave HCN, distinguishing features of mmWave communications, i.e. sensitivity to blockage and directional beamforming in transmitters and receivers, have been incorporated in the network model. The traditional approach to analyze the coverage for such networks is to ignore the holes and model the BSs as homogeneous independent PPP. This is considered as the baseline approach in our paper, and we compare our other proposed approaches with it. To enable downlink analysis of the SINR coverage probability, we first derive fairly accurate analytical expressions of the distance distribution of nearest LOS/NLOS BS to typical UE and association probability. The exact analysis of SINR coverage probability is not possible due to the PHP-based location of SBSs. Therefore, we take some different approaches to tackle this problem that we describe briefly in the following:

\begin{itemize}
	\item In the first approach, we ignore the effect of holes and approximate PHP by the baseline PPP of SBSs and analytical expressions for SINR distribution are derived in a mmWave HCN with a general LOS probability function. This approach actually is the same as when HCN is modeled by homogeneous independent PPP.
	Also as a corollary, we approximate PHP by PPP with the equivalent density. Surprisingly, this simple approximation yields a reliable analysis of mmWave HCN. 
	\item In the second approach, we consider the effect of holes when the typical UE is associated with an MBS and by incorporating the \textit{serving hole}, and an analytical expression for SINR coverage probability is derived. Besides, some new insights by considering the effect of a hole in SINR distribution have been discussed and its relation with circular holes have been discovered. Our results reveal that the analysis has been proposed in this paper, can provide a general study comparing with prior PHP-based modeling of the wireless networks.
	\item In the third approach, we incorporate the effect of nearest non-serving LOS and NLOS holes on SINR distribution. A special case of this analysis can be derived by incorporating nearest non-serving LOS or NLOS holes
	\item In the fourth approach, the effect of all non-serving holes are incorporated, but we ignore possible overlaps between them. Moreover, some interesting details about the effect of LOS and NLOS holes on SINR distribution have been perceived.
\end{itemize}

\subsection{Organization}
The rest of the paper is organized as follows. In section \ref{System Model}, considering distinctive features of mmWave communications, our system model is introduced. In section \ref{Coverage Probability}, distance distribution of nearest LOS/NLOS BS to typical UE, association probability and SINR coverage probability of the network based on different approaches are derived. In Section \ref{Numerical Results}, numerical results are presented to compare and validate different analytical expressions proposed in section \ref{Coverage Probability} and the impact of the system parameters on the performance evaluation are identified. Finally, conclusions and suggestion remarks for future works are provided in section \ref{Conclusion}.
\begin{table}
	\centering
	\caption{Notation and Description}
	\label{Notation and Description}
	\resizebox{.6\textwidth}{!}{
		\begin{tabular}{c|c}
			Symbol & Description \\ \hline
			$ \phi_{1}, \lambda_{1} $ & PPP distribution of MBSs, density of $ \phi_{1} $ \\ \hline
			$ \phi_{2}, \lambda_{2} $ & baseline PPP distribution of SBSs, density of $ \phi_{2} $\\ \hline 
			$ \psi, \lambda_{PHP} $ & PHP distribution of SBSs, density of $ \psi $ \\ \hline
			$ \phi_{u}, \lambda_{u} $ & PPP distribution of UEs, density of $ \phi_{u} $ \\ \hline
			$ BS_{k,0} $ & serving BS \\ \hline
			$ P_{k} $ & transmission power of BS in $ k^{th} $ tier \\ \hline
			$ G_{k} $ & directivity gain of BS in $ k^{th} $ tier \\ \hline
			$ \theta_{k} $ & beamwidth of BS in $ k^{th} $ tier \\ \hline
			$ \theta_{UE} $ & beamwidth of UE \\ \hline
			$ |h_{j,i}|^{2} $ & channel gain between the typical UE and the $ i^{th} $ BS in the $ j^{th} $ tier \\ \hline
			$ \alpha^{s} $  & $ s\in \{LOS,NLOS\} $  path loss exponent \\ \hline
			$ \upsilon^{s} $  & $ s\in \{LOS,NLOS\} $  Nakagami parameter \\ \hline
			$ P^{s}(r) $ & $ s \in \{LOS,NLOS\} $ probability \\ \hline
			$ P_{C} $ & SINR coverage probability \\ \hline
			$ S(x,D,\theta_{c}) $ & circular sector with radius $ D $ and central angle $ \theta_{c} $ centered at $ x $
		\end{tabular}
	}
\end{table}

\section{System Model}\label{System Model}
In this section, we describe our system model for evaluating coverage performance of a two-tier mmWave HCN. A brief overview of notations used in the system model and their descriptions have been summarized in Table \ref{Notation and Description}.
\subsection{BS and UE Locations}
The BSs in all tiers are operating in mmWave band with transmit power $ P_{k} $, $ k=1,2 $. The location of MBSs are distributed as a homogeneous PPP of density $ \lambda_{1} $ and the location of SBSs are treated as a PHP. Each MBS has an exclusion region that is a circular sector centered at the location of the MBS, with radius $ D $, and central angle $ \theta_{c} $. A typical realization of MBSs and SBSs are shown in Fig. \ref{Fig1}. We elaborate more on the properties of a PHP below.
\begin{definition}[Poisson Hole Process]\label{definition 1}
	A PHP is constructed by two independent PPPs: $ \phi_{1}\equiv\{x\} \subset \mathbb{R}^{2}$  of density $ \lambda_{1} $ and $ \phi_{2} $ of density $ \lambda_{2} $. The first process represents the hole centers, i.e. locations of MBSs, and the second process represents the baseline process for SBSs before excluding those who fall in the holes. For each $ x\in\phi_{1}$ we remove all the points in $ \phi_{2} \bigcap S(x,D,\theta_{c}) $, where $ S(x,D,\theta_{c}) $ is a circular sector with radius $ D $ and central angle $ \theta_{c} $ centered at the location of the points in $ \phi_{1} $. Then, the remaining points in $ \phi_{2} $ form the PHP $ \psi $, which can be mathematically defined as
	\begin{equation}
	\psi = \{ y\in\phi_{2}:y\notin \bigcup_{x\in \phi_{1}} S(x,D,\theta_{c})\}
	\end{equation}
\end{definition}
\begin{figure}
	\centering
	\includegraphics[width=0.6\textwidth]{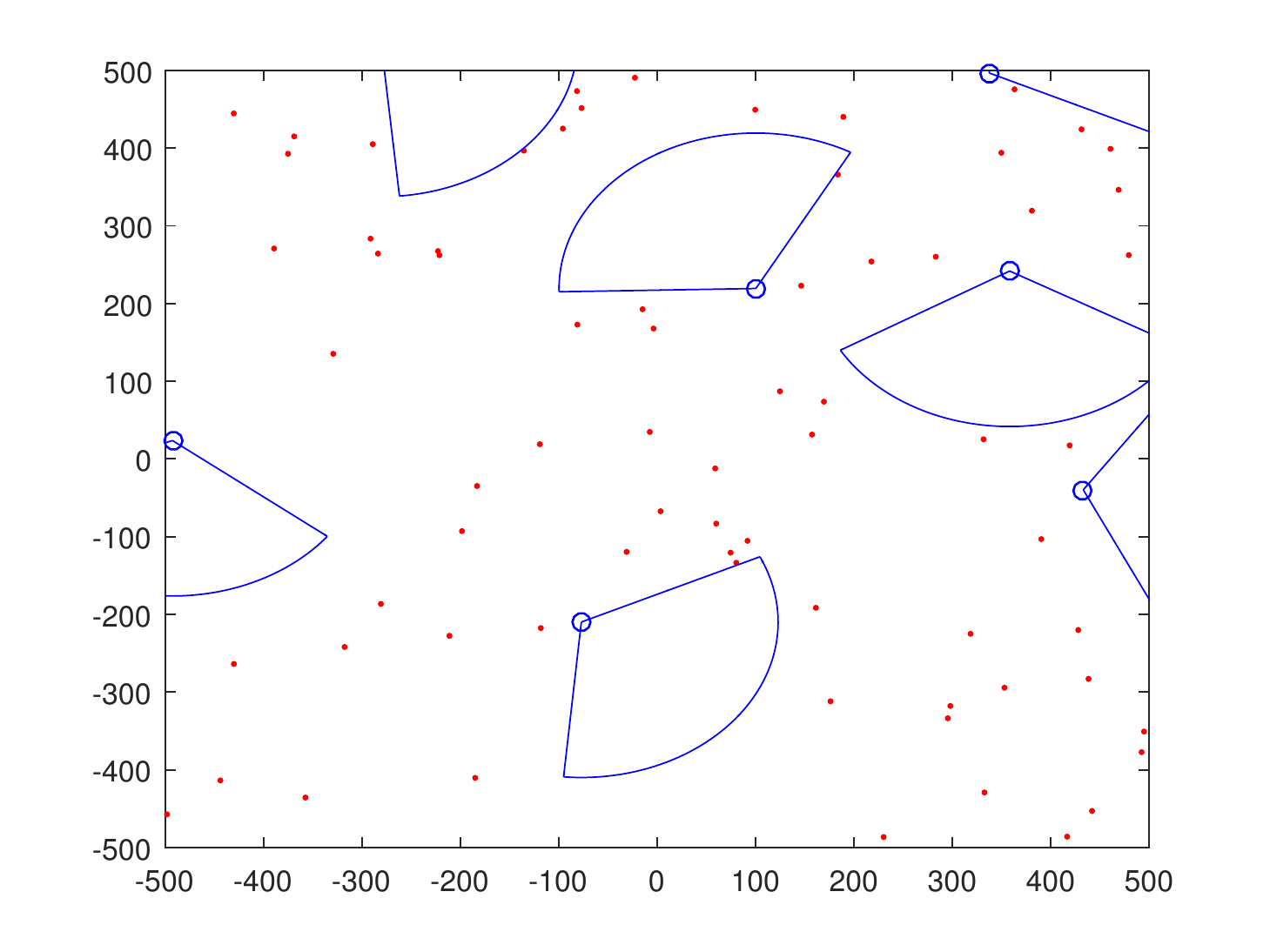}
	\caption{\small Two-tier HCN with spatial separation between tiers. The blue circles are the MBSs and circular sectors are exclusion regions with radius $ D $ and central angle $ \theta_{c} $. The red dots are SBSs deployed outside the exclusion regions.}
	\label{Fig1}
\end{figure}
The UEs are spatially distributed according to another homogeneous PPP and independent of BSs process, which is denoted by $ \phi_{u} $ of density $ \lambda_{u} $. The typical UE is assumed to be at the origin. Due to Slivnyak's theorem, this assumption does not change the statistical properties of the desired distribution \cite{Haenggi:2012:SGW:2480878}.
\subsection{Blocking}
Recent experimental results have revealed a high sensitivity of received  signals to blockage at mmWave frequencies. Therefore, the path loss in LOS and NLOS links are quite different. So, theoretical works in \cite{6932503}, \cite{7110547} have used different path loss exponents to model LOS and NLOS channels. Due to irregular building deployments, the links in the network can be either LOS or NLOS randomly, which can be modeled as a stochastic blockage proposed in \cite{6840343} based on random shape theory. According to this model, a LOS probability is defined as a deterministic non-increasing function of distance denoted as $ P^{LOS}(r)$, where $ r $ is the distance between BSs and the typical UE. For mathematical simplicity, different approximations have been used for LOS probability $ P^{LOS}(r) $. For example, in \cite{6932503} a LOS ball model is introduced, where LOS probability is equal to one inside of a ball and it is zero outside of the ball. There are also more generalized models in some papers such as \cite{7335646}, where piece-wise LOS probability function and multi-ball models are used. In this paper, we employ the general LOS probability function $ P^{LOS}(r) $ for our analysis.
\subsection{Beamforming}
Antenna arrays at MBSs, SBSs, and UEs are assumed to perform directional beamforming. A desirable advantage of such deployment is interference isolation, which reduces the impact of the intercell interference. The typical UE and its serving BS are assumed to have perfect channel state information (CSI) and can adjust the angles of the departure and arrival to achieve the maximum array gains. The beam direction of the interfering links is independently and uniformly distributed in the range $ [0,2\pi) $. We approximate the actual array pattern by the sectored antenna model introduced in \cite{6932503}. Based on this model, the directivity gain $ G_{k} $ in each tier is a discrete random variable with the value $ A_{k,g} $ and probability $ P_{k,g} $ $ (g=1,2,3,4) $, where $ A_{k,g} $ and $ P_{k,g} $ are constants defined in Table \ref{Probability mass function of directivity} and
\begin{table}
	\centering
	\caption{Probability mass function of directivity}
	\label{Probability mass function of directivity}
	\resizebox{0.6\textwidth}{!}{
		\begin{tabular}{|r|c|c|c|c|}
			\hline
			g&1&2&3&4\\ \hline
			$ A_{k,g} $& $ M_{k}M_{UE} $&$ M_{k}m_{UE} $&$ m_{k}M_{UE} $&$ m_{k}m_{UE} $\\ \hline
			$ P_{k,g} $&
			$ c_{k}c_{UE} $&$ c_{k}(1-c_{UE}) $&$ (1-c_{k})c_{UE} $&$ (1-c_{k})(1-c_{UE}) $\\
			\hline
		\end{tabular}
	}
\end{table}
$ c_{k}=\frac{\theta_{k}}{2\pi} $, $ c_{UE}=\frac{\theta_{UE}}{2\pi} $.
Sectored antenna model is the most traditional and tractable approximation for actual array beam pattern that has been adopted in \cite{6932503}, \cite{7110547}, \cite{7105406}, \cite{7931577}.
\subsection{Path Loss and Fading}
As discussed in blocking model, different path loss exponents are applied to LOS and NLOS links. Therefore, we define the path loss for each tier as below:
\begin{equation}
L(r)=
\left\{
\begin{array}{ll}
r^{-\alpha^{LOS}} & w.p.\quad P^{LOS}(r)\\
r^{-\alpha^{NLOS}} & w.p.\quad P^{NLOS}(r)
\end{array}
\right.
\label{pathloss}
\end{equation}
where $ \alpha^{LOS} $ and $ \alpha^{NLOS} $ are LOS and NLOS path loss exponents and $ P^{NLOS}(r) = 1-P^{LOS}(r) $. There are other path loss models such as \cite{7061455}, in which multi-slope path loss models are proposed. However, we use the standard path loss model given in (\ref{pathloss}) in this paper.

Due to minor small-scale fading in mmWave cellular systems, we assume independent Nakagami fading for each link, which is more general than Rayleigh fading used to model fading in conventional cellular networks. Also, different parameters of Nakagami fading $ \nu^{LOS} $ and $ \nu^{NLOS} $ are assumed for LOS and NLOS links and they are assumed to be positive integers for simplicity. Based on Nakagami fading assumption, the channel gain between the UE and BSs are normalized Gamma random variables. Shadowing is ignored in our system model but can be incorporated similar to \cite{7110547}.
\subsection{Cell Association}
We consider the following cell association rule: The typical UE is associated with the strongest BS in terms of long-term averaged received power at the typical UE. We assume that the typical UE is allowed to access either an MBS or an SBS. Therefore, the typical UE is associated with a $ s\in\{LOS,NLOS\} $ BS in $ k^{th} $ tier if and only if
\begin{equation}
P_{k}G_{k}r_{k}^{-\alpha^{s}}> P_{j}G_{j}r_{j}^{-\alpha^{s'}} 
\end{equation}
where $ P_{k} $, $ G_{k} $, $ r_{k}^{-\alpha^{s}} $, $ r_{j}^{-\alpha^{s'}} $ denote the transmission power, directivity gain in the serving link and path loss in $ k^{th} $ and $ j^{th} $ tier for ($ k,j=1,2$) and state $ s,s'\in\{LOS,NLOS\} $. It is worth noting that there are other cell association schemes used in cellular networks, such as instantaneous power based cell association or average biased-power cell association \cite{6287527}, that can be used similarly in our developed network analysis framework.

\section{SINR Coverage Probability}\label{Coverage Probability}

In this section, we evaluate the coverage probability of a two-tier HCN under the system model given in Section \ref{System Model} and derive expressions for SINR distribution in a mmWave network with general LOS probability function $ P^{LOS}(r) $.

In a multi-tier scenario, the serving BS may not always be the nearest BS for a given UE. However, it is obvious that serving BS is one of the closest BS from a given tier. We first derive the probability density function (PDF) of the distance separating the typical UE and the closest LOS/NLOS BS in each tier, which will be used in the upcoming analysis.

\begin{lemma}[Distance to the nearest $ s\in \{LOS,NLOS\} $ BS in $ k^{th} $ tier]\label{lemma_1}
	Given that the typical UE observes at least one $ s\in \{LOS,NLOS\} $ BS in $ k^{th} $ tier, the conditional PDF of distance to the nearest $ s\in \{LOS,NLOS\} $ BS in $ k^{th} $ tier is
	
	\begin{equation}
	f_{k}^{s}(r)= \varLambda'^{s}_{k,PHP}([0,r))exp(-\varLambda^{s}_{k,PHP}([0,r)))/B_{k}^{s}
	\end{equation}
	where
	\begin{equation}
	\varLambda^{s}_{k}([0,r))= 2\pi \lambda_{k}\int_{0}^{r} xP^{s}(x)dx,\quad \varLambda'^{s}_{k}([0,r))=\dfrac{d\varLambda^{s}_{k}([0,r))}{dr}
	\end{equation}
	\begin{equation}
	\varLambda_{1,PHP}^{s}([0,r)) = \varLambda_{1}^{s}([0,r)),\quad \varLambda_{2,PHP}^{s}([0,r)) = 2\pi \lambda_{PHP}\int_{0}^{r} xP^{s}(x)dx
	\end{equation}
	$ B_{k}^{s}=1-exp(-\varLambda^{s}_{k,PHP}([0,\infty))) $ is the probability that the typical UE observes at least one $ s\in \{LOS,NLOS\} $ BS in $ k^{th} $ tier and  $ \lambda_{PHP} = \lambda_{2}exp(-\frac{\lambda_{1} \theta_{c} D^{2}}{2}) $ is the equivalent density of PHP.
	
\end{lemma} 
\begin{proof}
	The equivalent density of PHP can be easily obtained similar to \cite{7557010} and the proof of PDF is similar to \cite{6932503} and can be derived using the null probability of a 2-D Poisson process and replacing $ \lambda_{2} $ by $ \lambda_{PHP} $.
\end{proof}
Next, we calculate the association probability given in lemma \ref{lemma-A_{k}^{s}}.
\begin{lemma}[Association probability]\label{lemma-A_{k}^{s}}
	The probability that the typical UE is associated with a $ s\in \{LOS,NLOS\} $ BS in $ k^{th} $ tier for $ k = 1,2 $ is
	\begin{equation}
	A_{k}^{s}= \int_{0}^{\infty} \varLambda'^{s}_{k,PHP}([0,x))
	exp\Bigg(-\sum_{j=1}^{2} \sum_{s'\in\{LOS,NLOS\}} \varLambda^{s'}_{j,PHP}([0,R_{j}^{s'}(x)) ) \Bigg)dx 
	\end{equation}
	where
	\begin{equation}
	R_{j}^{s'}(x)=(\frac{P_{j}}{P_{k}}\times \frac{M_{j}}{M_{k}})^{1/\alpha^{s'}}\times x^{\alpha^{s}/\alpha^{s'}} 
	\end{equation}
	is the exclusion radius of the  $ s'\in \{LOS,NLOS\} $ interferer BSs in $ j^{th} $ tier, when typical UE is associated with a $ s\in \{LOS,NLOS\} $ BS in $ k^{th} $ tier.
\end{lemma}

\begin{proof}
	See Appendix \ref{proof of association probability}.
\end{proof}

The accuracy of analytical expression in Lemma \ref{lemma_1} and Lemma \ref{lemma-A_{k}^{s}} will be demonstrated in section \ref{Numerical Results}.
Now we are prepared to develop a theoretical framework to analyze the downlink coverage probability for the typical UE. Based on the system model proposed in section \ref{System Model}, SINR received by the typical UE when the typical UE is associated with a $ s\in \{LOS,NLOS\} $ BS in $ k^{th} $ tier can be expressed as
\begin{equation}\label{SINR_{k}^{s}}
SINR_{k}^{s} = \frac{P_{k} |h_{k,0}|^{2} G_{k,0}r_{k}^{-\alpha^{s}}}{\sigma^{2} + I}
\end{equation}
where $ \sigma^{2} $ is thermal noise power and $ |h_{j,i}|^{2} $ is the channel gain between the typical UE and $ i^{th} $ BS in $ j^{th} $ tier, and $ I $ is the cumulative interference that the typical UE experience from all the other BSs, which is
\begin{equation}
	I= \sum_{s'\in{\{LOS,NLOS\}}} \sum_{i\in{\{\phi_{1}^{s'}\} \backslash BS_{k,0}}} P_{1} |h_{1,i}|^{2} G_{1,i}r_{1,i}^{-\alpha^{s'}}
	+\sum_{s'\in{\{LOS,NLOS\}}} \sum_{i\in{\{\psi^{s'}\} \backslash BS_{k,0}}} P_{2} |h_{2,i}|^{2} G_{2,i}r_{2,i}^{-\alpha^{s'}}
\end{equation}

The spatial distribution of the MBSs can be decomposed into two independent non-homogeneous process: the LOS process $ \phi_{1}^{LOS} $ and NLOS process $ \phi_{1}^{NLOS} $. Similarly, we have $ \psi^{LOS} $ and $ \psi^{NLOS} $ for PHP distribution of SBSs. Therefore, the SINR coverage probability $ (P_{C}) $ of the network can be computed using the total probability law as follows
\begin{equation}
P_{C}=Pr(SINR>\tau_{k}) = \sum_{k=1}^{2} \sum_{s\in\{LOS,NLOS\}} A_{k}^{s}P_{C_{k}}^{s}(\tau_{k})
\end{equation}
where $ P_{C_{k}}^{s}(\tau_{k}) $ and $ A_{k}^{s} $ is the conditional coverage and association probability given that the typical UE is associated with a $ s\in \{LOS,NLOS\} $ BS in $ k^{th} $ tier. In the following sections, we calculate coverage probability for a two-tier mmWave HCN using a few different approaches.

\subsection{Approximating PHP by baseline PPP}
In this approach, PHP distribution of SBSs is approximated by baseline PPP $ \phi_{2} $. This approach is the same as modeling SBSs as PPP, which ignores the correlation between the location of BSs of different tiers. Next theorem provides coverage probability based on this approximation.
\begin{theorem}\label{Theorem1}
	The PPP-based coverage probability for a two tier
	mmWave HCN is
	\begin{multline}\label{equation coverage_PPP}
	P_{C} \approx
	\mathlarger\sum_{k=1}^{2} \mathlarger\sum_{s\in\{LOS,NLOS\}} \mathlarger\sum_{n=1}^{\upsilon^{s}} (-1)^{n+1} \binom{\upsilon^{s}}{n} \mathlarger\int_{0}^{\infty} exp(-\mu_{k,n}^{s}\sigma^{2})\\
	exp\Bigg(- \mathlarger\sum_{j=1}^{2} \mathlarger\sum_{s'\in\{LOS,NLOS\}}
	W_{j}^{s'}(x) + \varLambda^{s'}_{j,PHP}([0,R_{j}^{s'}(x))) \Bigg)
	\varLambda'^{s}_{k,PHP}([0,x))dx
	\end{multline}
	where
	\begin{equation}\label{W_{j}^{s'}(x)}
	W_{j}^{s'}(x)=
	\mathlarger\sum_{g=1}^{4} P_{j,g} \mathlarger\int_{R_{j}^{s'}(x)}^{\infty}
	F\Bigg(\upsilon^{s'} , \frac{\mu_{k,n}^{s} P_{j}A_{j,g}r^{-\alpha^{s'}} }{\upsilon^{s'}}\Bigg) 
	\varLambda'^{s'}_{j}([0,r)) dr
	\end{equation}
	and $ F(v,x) = 1-(\frac{1} {1+x})^{v} $, $ \mu_{k,n}^{s} = \frac{n\tau_{k}\eta^{s}}{P_{k}G_{k,0}x^{-\alpha^{s}}} $, $ \eta^{s} = \nu^{s} (\nu^{s}!)^{-1/\nu^{s}} $.
	
\end{theorem}

\begin{proof}
	See Appendix \ref{proof of SINR coverage probability-PPP}.
\end{proof}
\begin{corollary}\label{corollary1}
	As a consequence, using independent thining of a PPP, we approximate PHP distribution of SBSs by a PPP with $ \lambda_{PHP} $. In this case, it is enough to replace $ \varLambda'^{s'}_{j}([0,r))$ with $ \varLambda'^{s'}_{j,PHP}([0,r)) $ in the equation \ref{W_{j}^{s'}(x)} in Theorem \ref{Theorem1}.
\end{corollary}

\subsection{Incorporating serving hole}
To provide more accurate coverage probability, we incorporate only the effect of the serving hole in our analysis. The term serving hole used in this paper represents the hole that is associated with the serving MBS. Since serving MBS is the closest LOS/NLOS MBS to the typical UE, serving hole is the nearest LOS/NLOS hole to the typical UE. However, in the case when the typical UE is associated with a LOS MBS, it is possible that some non-serving NLOS MBSs exist who are closer to the typical UE, hence, their holes are closer than serving hole as well.
\begin{theorem}\label{Theorem3}
	The coverage probability by incorporating serving hole can be evaluated as
\begin{multline}\label{theorem2}
	P_{C} \approx
	\mathlarger\sum_{k=1}^{2} \mathlarger\sum_{s\in\{LOS,NLOS\}} \mathlarger\sum_{n=1}^{\upsilon^{s}} (-1)^{n+1} 	\binom{\upsilon^{s}}{n} \mathlarger\int_{0}^{\infty} exp(-\mu_{k,n}^{s}\sigma^{2})
	exp\Bigg(- \mathlarger\sum_{j=1}^{2} \mathlarger\sum_{s'\in\{LOS,NLOS\}}\\
	W_{j}^{s'}(x) + \varLambda^{s'}_{j,PHP}([0,R_{j}^{s'}(x))) \Bigg)
	exp\Bigg(\mathlarger\sum_{s'\in\{LOS,NLOS\}} Q^{s'}(x) \mathbbm{1}(k=1)\Bigg) \varLambda'^{s}_{k,PHP}([0,x)) dx
\end{multline}
	
	where
\begin{equation}\label{Q}
\begin{split}
		Q^{s'}(x) =& \frac{\theta_{c}}{2\pi} \lambda_{2} \mathlarger\sum_{g=1}^{4} P_{2,g} \mathlarger\int_{0}^{2\pi} \mathlarger\int_{0}^{D} F\Bigg(\upsilon^{s'} , \frac{\mu_{k,n}^{s} P_{2}A_{2,g}(u^{2}+x^{2}-2uxcos(\varphi))^{-\alpha^{s'}/2} }{\upsilon^{s'}}\Bigg) \\
		&P^{s'}((u^{2}+x^{2}-2uxcos(\varphi))^{1/2}) udud\varphi\\
		=&\mathlarger\sum_{g=1}^{4} P_{2,g} \mathlarger\int_{x-D}^{x+D} F\Bigg(\upsilon^{s'} , \frac{\mu_{k,n}^{s} P_{2}A_{2,g}u^{-\alpha^{s'}} }{\upsilon^{s'}}\Bigg) 2\pi\lambda(u)
		P^{s'}(u) udu
\end{split}
\end{equation}
	
where $ \lambda(u) = \lambda_{2}\Bigg(\frac{\theta_{c}}{2\pi}\times \frac{\arccos(\frac{u^2+x^2-D^2}{2ux})}{\pi}\Bigg) $ and $ \mathbbm{1}(.) $ is the indicator function.
\end{theorem}

\begin{proof}
	See Appendix \ref{Proof of SINR Coverage Probability-PHP_serving hole}.
\end{proof}

\begin{remark}
	The above analytical expressions reveal that the equivalent area of circular sector holes in random direction with radius D and central angle $ \theta_{c} $ is $ \frac{\theta_{c}}{2\pi} \times$ area of circular holes with radius D. So, the Laplace transform of interference distribution with considering circular sector holes in random direction appears approximately as $ exp(\frac{\theta_{c}}{2\pi})\times L_{I_{C}} $, where $  L_{I_{C}} $ is the Laplace transform of interference distribution with circular holes. A more interesting interpretation of the above result is that PHP distribution of SBSs with a circular sector hole in random direction and at distance $ x $ is approximately a non-homogenous PPP with density $ \lambda_{2}\Bigg(1-\frac{\theta_{c}}{2\pi}\frac{\arccos(\frac{u^2+x^2-D^2}{2ux})}{\pi}\Bigg) $. This result is similar to \cite{7557010}, which is a special case of ours.  Clearly, the computational complexity by such setup of holes is not increased compared to traditional circular holes. Moreover, it is worth noting that the two-fold integral in equation \ref{Q} can be applied to any hole shape. To summarize, the analysis in this paper can provide a more general study comparing with prior PHP modeling of the wireless networks.
\end{remark}

\subsection{Incorporating nearest non-serving LOS and NLOS holes}
By ignoring all the holes except the nearest non-serving LOS and NLOS ones, we evaluate the coverage performance of the network. In other words, in this approach, we consider the nearest holes that are associated with the interferer MBSs, i.e. the MBSs who are not serving the typical UE, and therefore, the effect of the serving hole has not taken into account in this approach. Note that the term LOS/NLOS hole used in this paper represents the hole that is associated with the LOS/NLOS MBS. 
\begin{theorem}\label{Theorem4}
	The coverage probability by incorporating the nearest non-serving LOS and NLOS holes can be expressed as
	\begin{multline}
	P_{C} \approx
	\mathlarger\sum_{k=1}^{2} \mathlarger\sum_{s\in\{LOS,NLOS\}} \mathlarger\sum_{n=1}^{\upsilon^{s}} (-1)^{n+1} 		\binom{\upsilon^{s}}{n} \mathlarger\int_{0}^{\infty} exp(-\mu_{k,n}^{s}\sigma^{2})\\
	exp\Bigg(- \mathlarger\sum_{j=1}^{2} \mathlarger\sum_{s'\in\{LOS,NLOS\}}
	W_{j}^{s'}(x) + \varLambda^{s'}_{j}([0,R_{j}^{s'}(x))) \Bigg)
	  Z(x) \varLambda'^{s}_{k}([0,x)) dx
	\end{multline}
	where
	\begin{multline}
	Z(x)=
	\mathlarger\prod_{s''\in\{{LOS,NLOS}\}}\mathlarger\int_{R_{1}^{s''}(x)}^{\infty} 
	exp \Bigg(\mathlarger\sum_{s'\in\{{LOS,NLOS}\}} Q^{s'}(y)\Bigg)\\
	\frac{\varLambda'^{s''}_{1}([0,y)) exp\Bigg(-\varLambda_{1}^{s''}([0,y))+\varLambda_{1}^{s''}([0,R_{1}^{s''}(x)))\Bigg)}
	{1-exp\Bigg(-\varLambda_{1}([R_{1}^{s''}(x),\infty))\Bigg)} dy
	\end{multline}
\end{theorem}
\begin{proof}
	See Appendix \ref{Proof of SINR Coverage Probability-PHP_nearest_holes}.
\end{proof}

\subsection{Incorporating all non-serving holes}
In this approach, we provide an analytical expression for the coverage probability by incorporating all non-serving holes. However, here the overlaps between them are ignored. Due to possible overlaps among holes, there are some points that will be removed multiple times. Note that similar to Theorem \ref{Theorem4}, the effect of the serving hole has not been considered in this approach.
\begin{theorem}\label{Theorem5}
	The coverage probability by incorporating all non-serving holes is
	\begin{multline}\label{eq-all-holes}
	P_{C} \approx
	\mathlarger\sum_{k=1}^{2} \mathlarger\sum_{s\in\{LOS,NLOS\}} \mathlarger\sum_{n=1}^{\upsilon^{s}} (-1)^{n+1} 		\binom{\upsilon^{s}}{n} \mathlarger\int_{0}^{\infty} exp(-\mu_{k,n}^{s}\sigma^{2})\\
	exp\Bigg(- \mathlarger\sum_{j=1}^{2} \mathlarger\sum_{s'\in\{LOS,NLOS\}}
	W_{j}^{s'}(x) + \varLambda^{s'}_{j}([0,R_{j}^{s'}(x))) \Bigg)
	T(x) \varLambda'^{s}_{k}([0,x)) dx
	\end{multline}
	where
	\begin{equation}\label{T(x)}
	T(x) = exp\Bigg(- \mathlarger\sum_{s''\in\{{LOS,NLOS}\}}\mathlarger\int_{R_{1}^{s''}(x)}^{\infty} \mathlarger\sum_{s'\in\{{LOS,NLOS}\}}  \Bigg(1-exp(Q^{s'}(y))\Bigg)\varLambda_{1}^{s''}([0,y))dy\Bigg)
	\end{equation}
\end{theorem}

\begin{proof}
	See Appendix \ref{Proof of SINR Coverage Probability-PHP_all_holes}.
\end{proof}
\begin{remark}\label{Remark2}
	The above analytical expressions in Theorem \ref{Theorem4} and \ref{Theorem5} incorporates both LOS and NLOS holes. A special case can be obtained by incorporating only LOS or NLOS holes. The interesting fact about LOS and NLOS holes is that NLOS holes have a dominant effect on the received signal. This is because of two reasons: first, in the dense building environment like urban areas, the density of NLOS links are much more than LOS links. Therefore, the average number of NLOS holes are much more than LOS holes and LOS holes have less impact on SINR distribution than NLOS holes.. 
	To expand mathematically on this claim, we derive the ratio of the average number of NLOS BSs to LOS BSs,
	\begin{equation*}
	\begin{split}
	\frac{\rho^{NLOS}}{\rho^{LOS}} &= \frac{2\pi\lambda\int_{0}^{\infty}r(1-P^{LOS}(r))dr}{2\pi\lambda\int_{0}^{\infty}rP^{LOS}(r)dr}
	\end{split}
	\end{equation*}
	if $ \int_{0}^{\infty}rP^{LOS}(r)dr<\infty $, $ \frac{\rho^{NLOS}}{\rho^{LOS}} = \infty $. But if $ \int_{0}^{\infty}rP^{LOS}(r)dr<\infty $ is not satisfied,
	\begin{equation*}
	\begin{split}
	\frac{\rho^{NLOS}}{\rho^{LOS}} &= \frac{2\pi\lambda\int_{0}^{\infty}r(1-P^{LOS}(r))dr}{2\pi\lambda\int_{0}^{\infty}rP^{LOS}(r)dr} = \frac{\int_{0}^{\infty}rdr}{\int_{0}^{\infty}rP^{LOS}(r)dr} - 1 = \lim_{x\to\infty}\frac{\int_{0}^{x}rdr}{\int_{0}^{x}rP^{LOS}(r)dr} - 1 \\&= \lim_{x\to\infty} \frac{x}{xP^{LOS}(x)} - 1 = \lim_{x\to\infty} \frac{1}{P^{LOS}(x)} - 1
	\end{split}
	\end{equation*}
	If we assume that in long enough distances, $ P^{LOS}(x) $ is sufficiently close to $ 0 $, the above ratio goes to $ \infty $.
	Second, the exclusion radius of LOS MBSs is almost surely much bigger than NLOS BSs and Since received signal solely depends upon the distance between the transmitter and receiver, the effect of LOS holes on the received signal is less than the effect of NLOS holes. Mathematically,
	\begin{equation*}
	\frac{R_{1}^{LOS}(x)}{R_{1}^{NLOS}(x)} = \kappa^{\gamma}x^{\alpha^{s}\gamma}>1 \quad ,\quad for \quad  x>1
	\end{equation*}
	where $ \gamma = \frac{1}{\alpha^{LOS}} - \frac{1}{\alpha^{NLOS}}>0 $ and $ \kappa=\frac{P_{1}}{P_{k}}\times \frac{M_{1}}{M_{k}}\geq1 $. Note that for $ x<1 $, the above results may not be true but in the context of HCN, it rarely happens that distance between typical UE and serving BS be less than $ 1 $ meters. 
	To clarify more on this claim, please refer to Fig. \ref{Remark2_Fig}. This remark will be demonstrated numerically in the section \ref{Numerical Results}.
	\begin{figure}
		\centering
		\includegraphics[width=0.7\textwidth , height=0.45\textwidth]{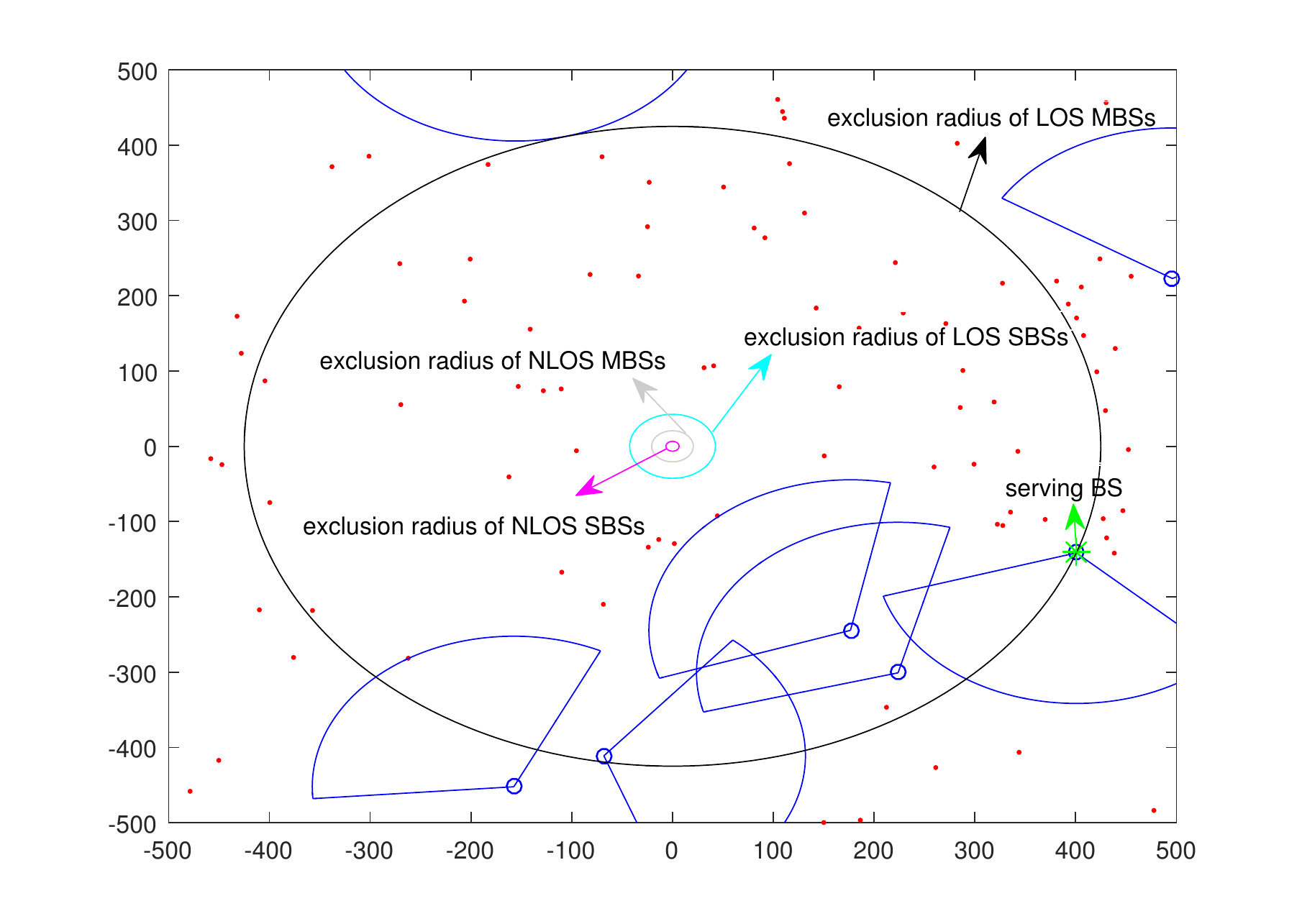}
		\caption{\small An instance of BSs location, holes and exclusion regions.}
		\label{Remark2_Fig}
	\end{figure}

	Another useful fact that can be concluded from Fig. \ref{Remark2_Fig} is that since we have $ R_{1}^{LOS}>>R_{2}^{s}, s\in\{LOS,NLOS\} $, the overlaps between the exclusion region of LOS/NLOS SBSs and LOS holes are negligible and ignoring these overlaps in deriving analytical expressions seems reasonable. However, this claim may not be true for NLOS holes, but due to relatively small exclusion region of LOS/NLOS SBSs, we have ignored the overlaps between these regions and NLOS holes as well. Note that considering the exact non-overlapped region strongly intensify the complexity of analytical expressions and make them even non-tractable. For example, as it can be seen in Fig. \ref{Remark2_Fig}, unlike the microwave two-tier cellular network with two exclusion region, there are four of them corresponding to LOS/NLOS BSs in two tiers and also have different radius by associating with LOS/NLOS BSs of different tiers. In section \ref{Numerical Results}, justification of this approximation will be evaluated.
\end{remark}

We now move to the next section, where the accuracy of our  approaches is evaluated by Monte Carlo simulations.

\section{Numerical Results}\label{Numerical Results}

In this section, we evaluate the analytical expressions using numerical integration and validate the accuracy of the proposed expressions by comparing the simulation results. We assume MBSs and SBSs are operating at 28 GHz and the bandwidth assigned to each UE is BW = 1 GHz. Power path loss law are the same for both tiers, but different for LOS and NLOS links. The LOS probability function is $ P^{LOS}(r) = exp(-\beta r) $, where $ \beta = \sqrt{2}/R^{LOS} $, $ R^{LOS}=200 $ meters, and $ R^{LOS} $ indicates average LOS distance. Two different Setups are considered which have been summarized as follow and in Table \ref{Setups}.
\begin{table}
	\centering
	\caption{Setups}
	\label{Setups}
	\resizebox{.6\textwidth}{!}{
		\begin{tabular}{|c|c|}
			\hline
			Setup 1 & Setup 2 \\ \hline
			$ \lambda_{1} = 2.5 $ $ \text{MBSs/km}^{2} $ & $ \lambda_{1} = 10 $ $ \text{MBSs/km}^{2} $ \\ \hline
			$ \lambda_{2} = 5\lambda_{1} $ $ \text{SBSs/km}^{2} $ & $ \lambda_{2} = 20\lambda_{1} $ $ \text{SBSs/km}^{2} $ \\ \hline
			$ D = 100 $ meters &  $ D = 200 $ meters\\ \hline
			$ \theta_{c}=\frac{\pi}{3} $ & $ \theta_{c}=\frac{2\pi}{3} $\\ \hline
			$ P_{2}=P_{1}-30 $ dB & $ P_{2}=P_{1}-20 $ dB \\ \hline
			$ M_{2}=M_{1} - 5 $ dB & $ M_{2}=M_{1} $\\ \hline
			$ \theta_{2} = \frac{\pi}{6} $ & $ \theta_{2} = \frac{\pi}{3} $\\
			\hline
		\end{tabular}
	}
\end{table}
\begin{itemize}
	\item  $ P_{1} = 53 $ dBm, $ M_{1} = 10 $ dB, $ \theta_{1} = \frac{\pi}{3} $. 
	\item  $ FBR_{k} = 20 $ dB $ m_{k}=M_{k}-FBR_{k} $ ($ k=1,2 $).
	\item $ M_{UE} = 10 $ dB, $ FBR_{UE} = 20 $ dB, $ m_{UE}=M_{UE}-FBR_{UE} $, $ \theta_{UE} = \frac{\pi}{2} $.
	\item $ \alpha^{LOS} = 2 $, $ \alpha^{NLOS} = 4 $.
	\item $ \upsilon^{LOS} = 3 $, $ \upsilon^{NLOS} = 2 $.
	\item $ \sigma^{2}= -174 $ dBm/Hz + $ 10\text{log}_{10}(\text{BW}) $ + $ 10 $ dB.
\end{itemize}

First, the evaluation of distance distribution of nearest LOS/NLOS SBS to typical UE and association probability have been 
provided in Figs. \ref{Distance_Distribution_Fig}, \ref{Association_Probability_Fig} based on parameters adjusted in Setup 2. Note that since MBSs are distributed as PPP, the analytical expression in Lemma \ref{lemma_1} is the true distance distribution of nearest LOS/NLOS MBS to typical UE. So we have only plotted the distance distribution of nearest LOS/NLOS SBS. The numerical results in Fig. \ref{Distance_Distribution_Fig} represent a fairly good match between the analytical expression of Lemma \ref{lemma_1} and simulation. This claim is true for association probability based on numerical and simulation results in Fig. \ref{Association_Probability_Fig}. Moreover, the results in Fig. \ref{Association_Probability_Fig} indicate that due to high power attenuation in NLOS links, typical UE rarely associates to NLOS BSs. Also, as it can be seen, in a dense blocking environment, typical UE is served more likely by LOS SBS and as the density of blocking are decreasing, the probability that typical UE is associated by LOS MBS is more than LOS SBS.

\begin{figure}
	\centering
	\includegraphics[width=0.7\textwidth , height=0.45\textwidth]{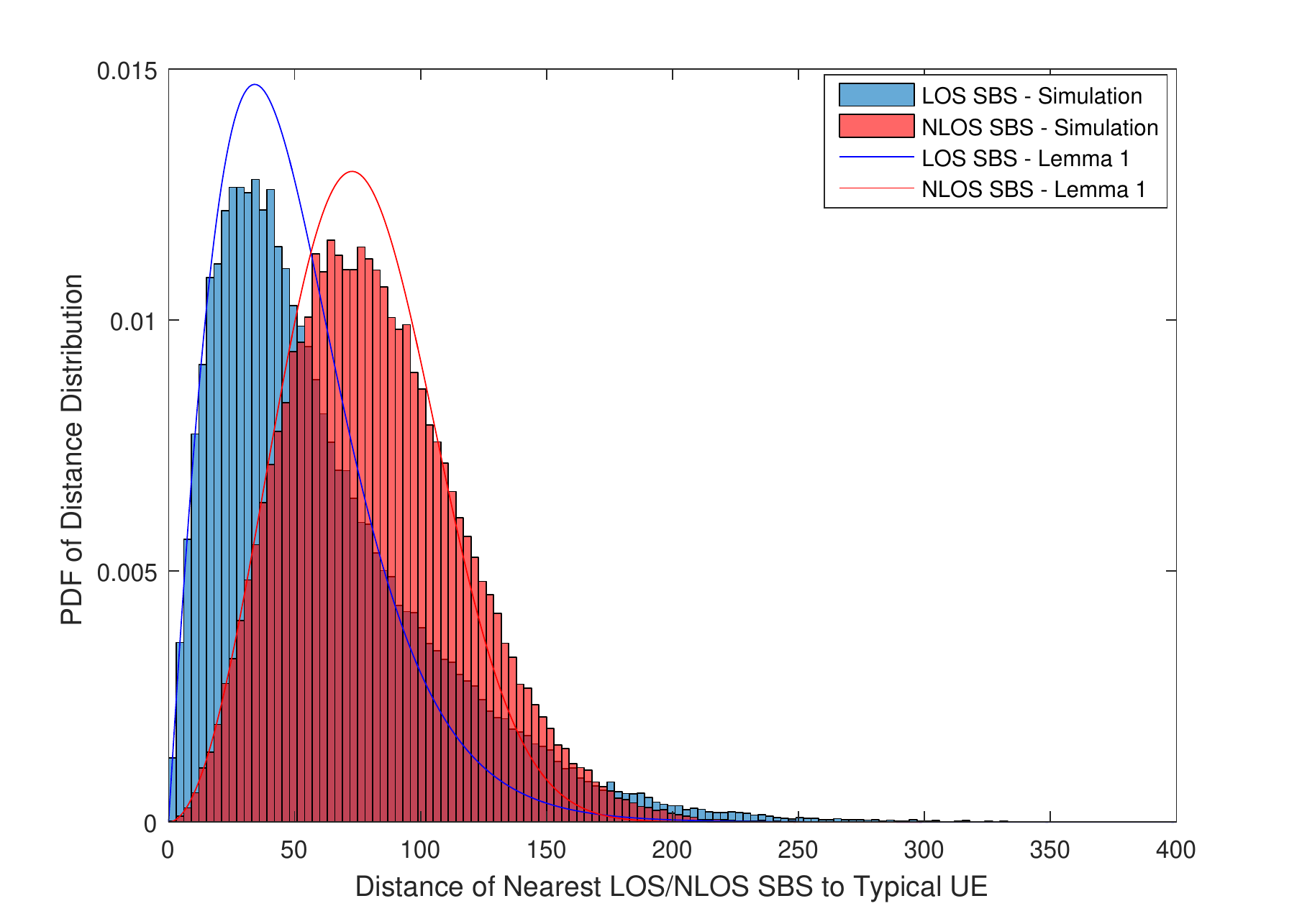}
	\caption{\small PDF of diststance distribution of nearest LOS/NLOS SBS to typical UE.}
	\label{Distance_Distribution_Fig}
\end{figure}
\begin{figure}
	\centering
	\includegraphics[width=0.7\textwidth , height=0.45\textwidth]{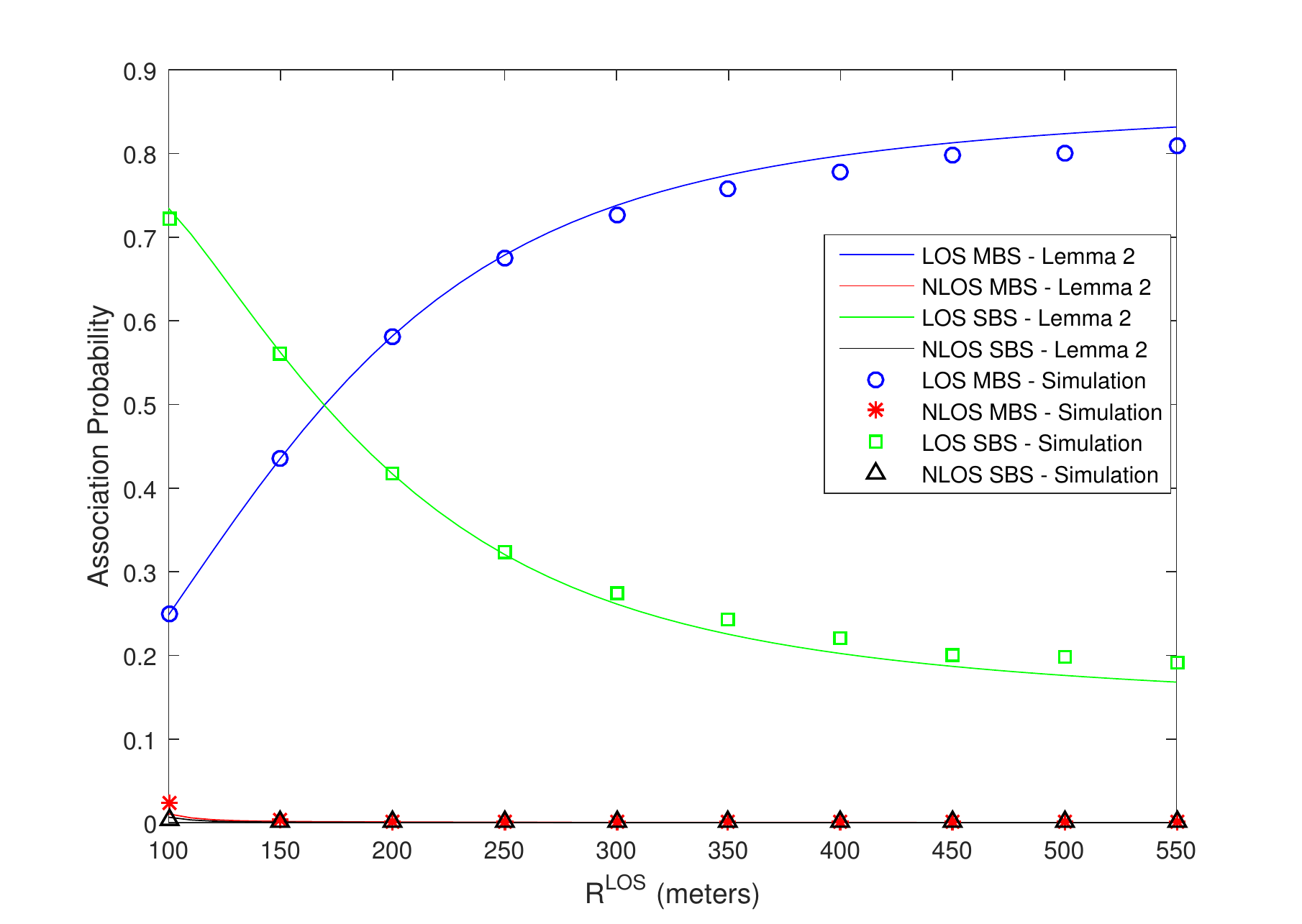}
	\caption{\small Association probability as a function of average LOS distance.}
	\label{Association_Probability_Fig}
\end{figure}

\begin{figure}
	\centering
	\includegraphics[width=0.7\textwidth , height=0.45\textwidth]{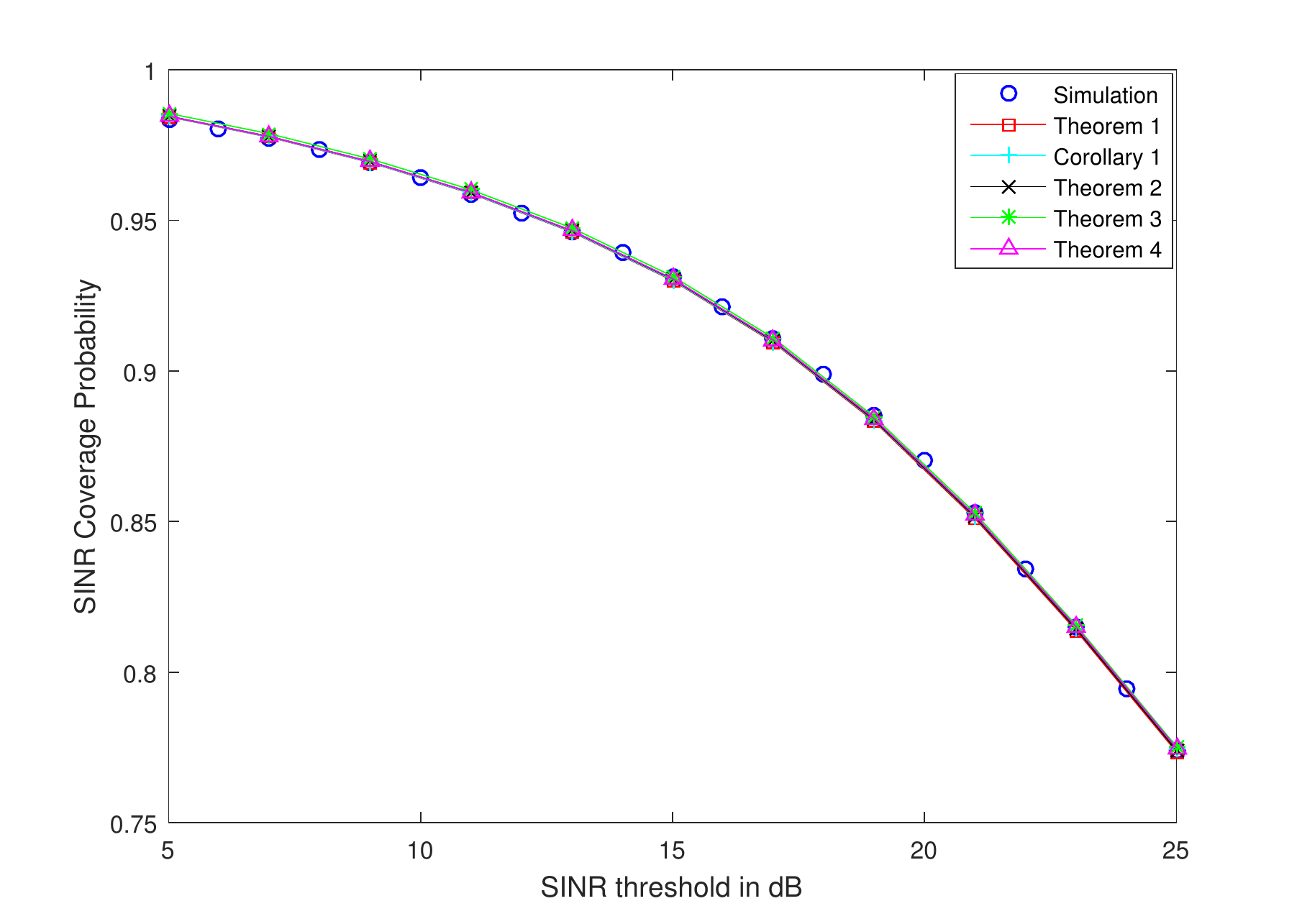}
	\caption{\small SINR coverage probability as a function of SINR threshold- Setup 1.}
	\label{setup1}
\end{figure}
\begin{figure}
	\centering
	\includegraphics[width=0.7\textwidth , height=0.45\textwidth]{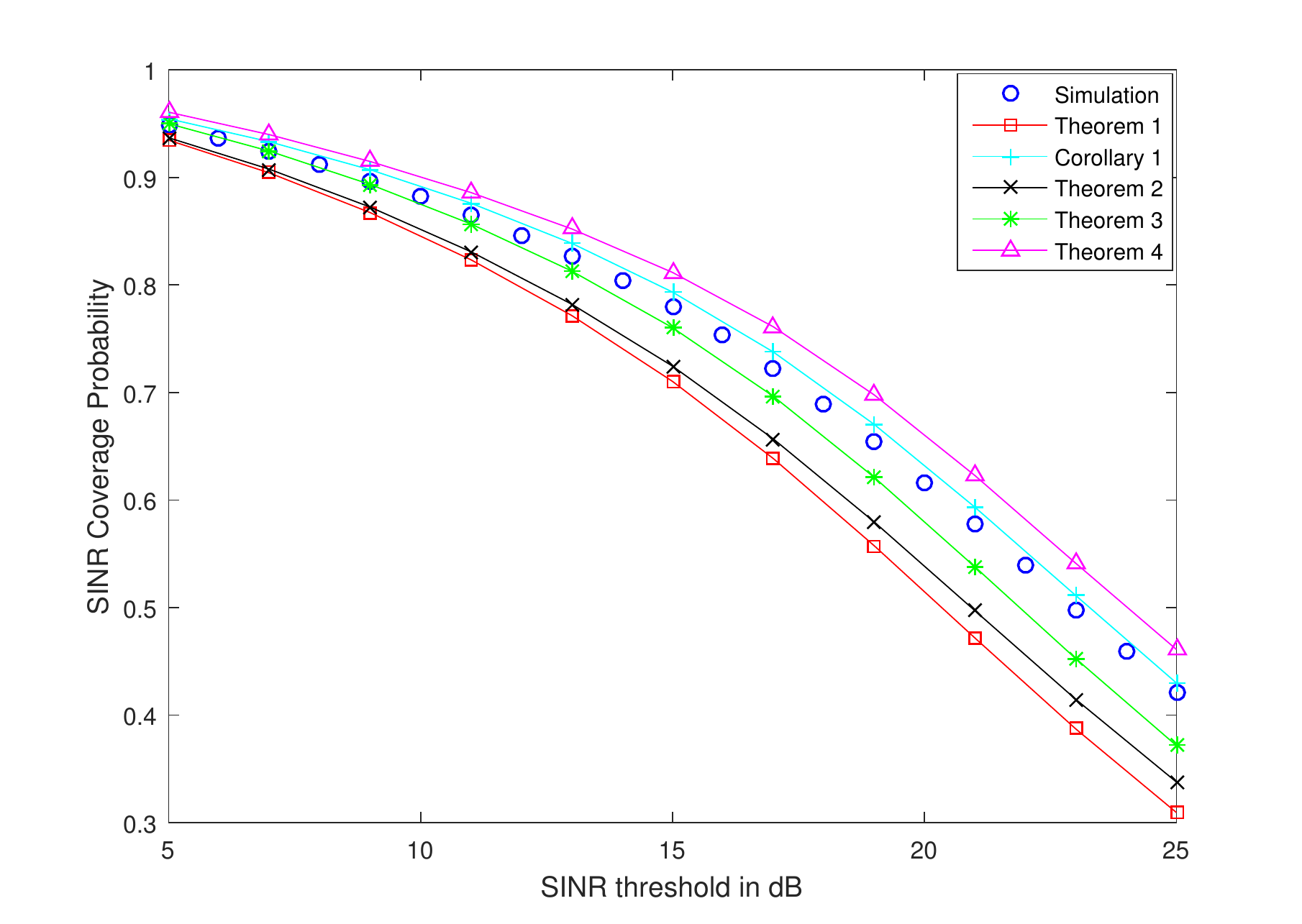}
	\caption{\small SINR coverage probability as a function of SINR threshold- Setup 2.}
	\label{setup2}
\end{figure}

Next, we plot the analytical curves and the simulation results for the SINR coverage probability as a function of SINR threshold for two Setups in Figs. \ref{setup1}, \ref{setup2}. Clearly, by setting system parameters as defined in Setup 1, the effect of holes has less impact on the received signal and all approaches have similar performance. On the other hand in the Setup 2, where the effect of holes is significant, analytical expressions of different approaches deviate from the simulation result. But all the proposed approaches in this paper have superior performance than PPP-based analysis. As it is clear in this figure, the effect of the serving hole on the SINR distribution is not significant. This observation is a direct result of Remark \ref{Remark2} and the fact that has been described in Fig. \ref{Association_Probability_Fig}. Also, it is worth noting that the result in Corollary \ref{corollary1}, unlike similar approach in \cite{7557010}, have been provided a reliable performance. The reason behind this difference is the network model that authors have been considered in \cite{7557010}, which is similar to a simple \textit{ad hoc} setup. So, in the context of mmWave HCN, this approach provides an easy to use and worthy analysis.
To testify the accuracy of the analytical expressions in Theorem \ref{Theorem1}-\ref{Theorem5} and Corollary \ref{corollary1} with circular hole shapes, all approaches have been plotted in Fig. \ref{Circular}, which again confirms the accuracy of our results. To authorize claims in Remark \ref{Remark2}, the analytical expression in Theorem \ref{Theorem5} have been plotted again in Fig. \ref{LOS_NLOS} by varying average LOS distance. The curves in this figure clearly prove that NLOS holes have a more dominant effect on SINR distribution than LOS holes, even in the very sparse building environment.
\begin{figure}
	\centering
	\includegraphics[width=0.7\textwidth , height=0.45\textwidth]{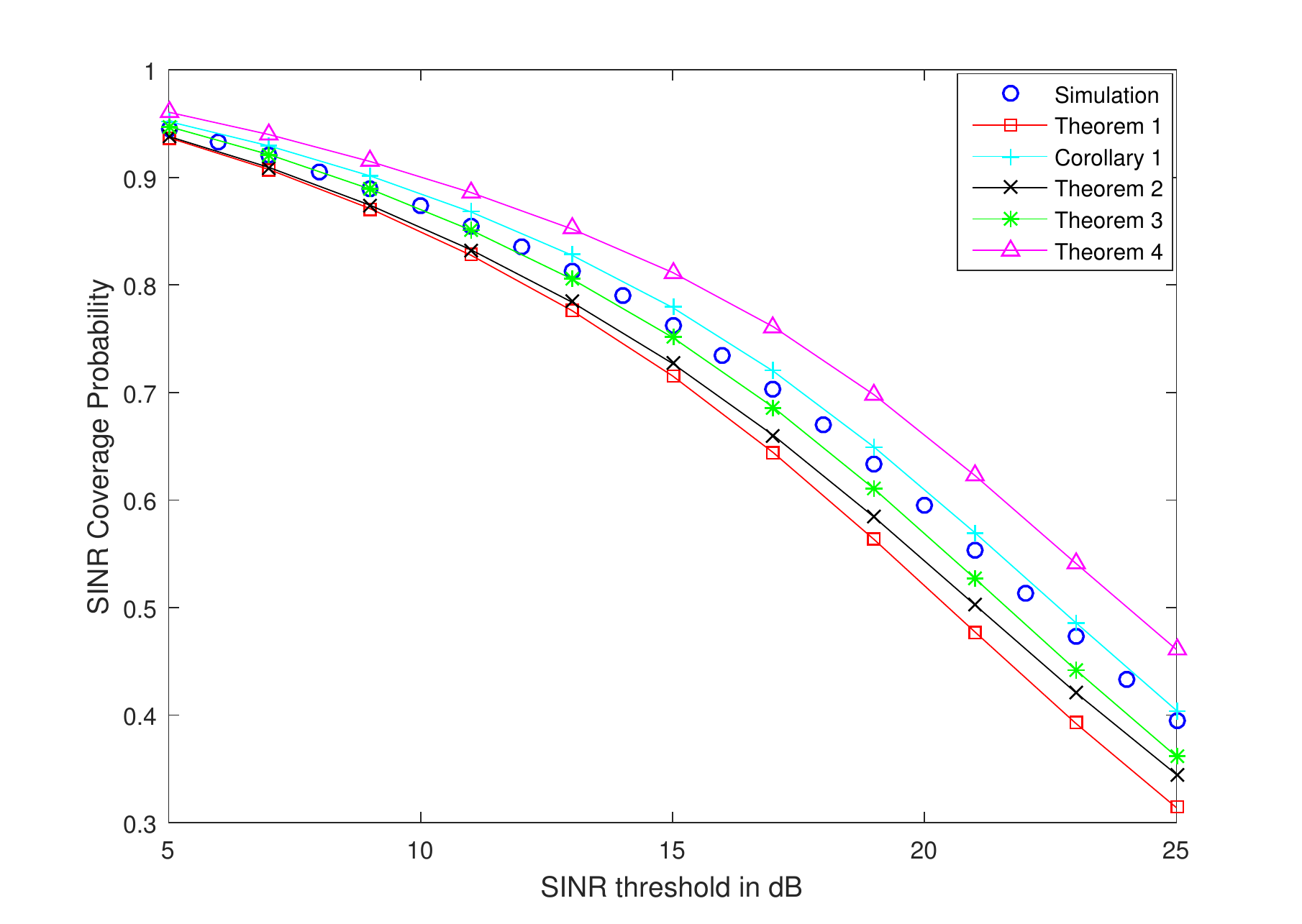}
	\caption{\small SINR coverage probability as a function of SINR threshold- $ \theta_{c} = 2\pi $, $ D=100 $ meters}
	\label{Circular}
\end{figure}
\begin{figure}
	\centering
	\includegraphics[width=0.7\textwidth , height=0.45\textwidth]{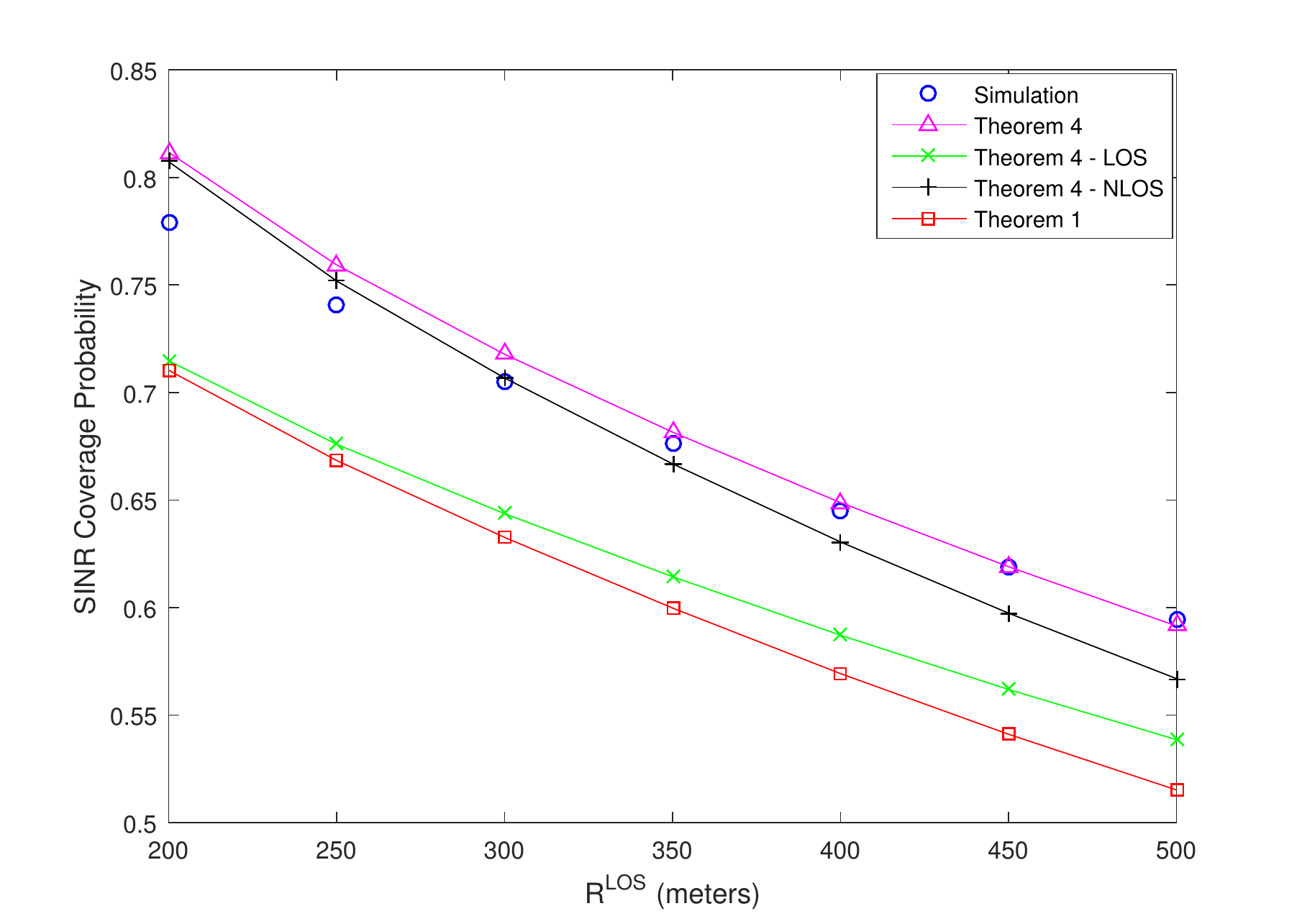}
	\caption{\small SINR coverage probability as a function of average LOS distance.}
	\label{LOS_NLOS}
\end{figure}

\begin{figure}
	\centering
	\includegraphics[width=0.7\textwidth , height=0.45\textwidth]{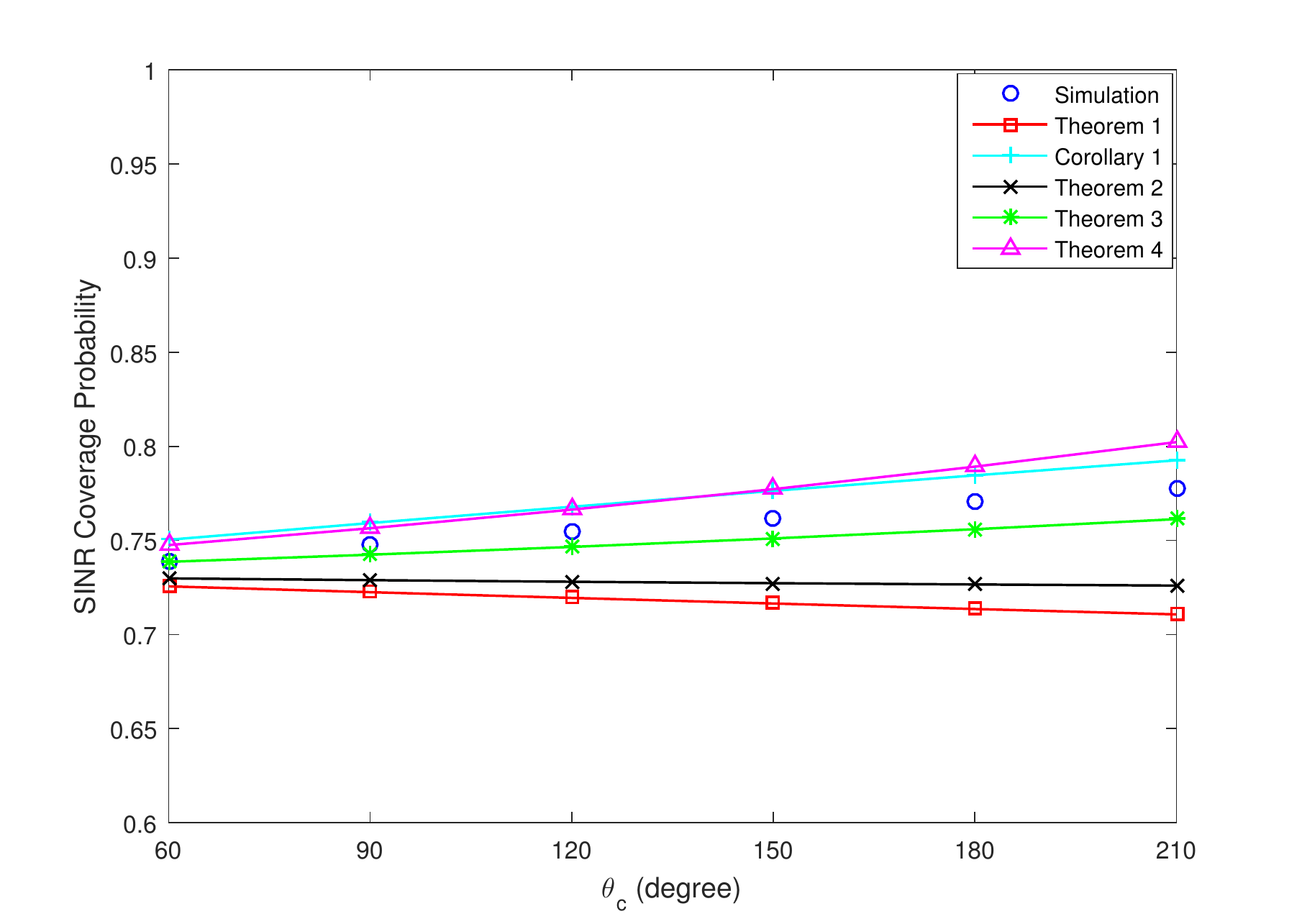}
	\caption{\small SINR coverage probability as a function of central angle of holes.}
	\label{theta_c_Fig}
	\hspace{-7mm}
\end{figure}
\begin{figure}
	\centering
	\includegraphics[width=0.7\textwidth , height=0.45\textwidth]{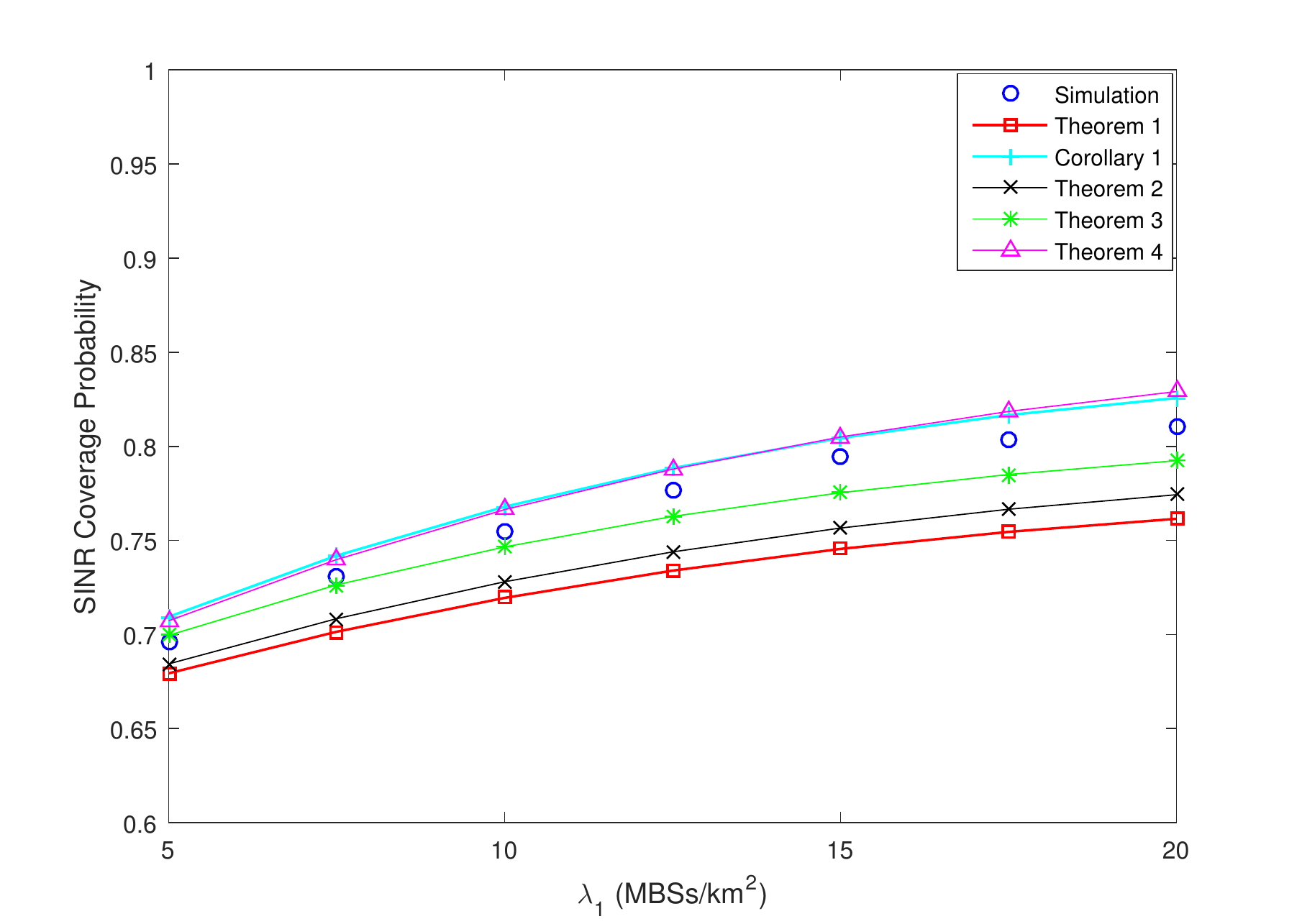}
	\caption{\small SINR coverage probability as a function of MBSs' density.}
	\label{lambda1_Fig}
\end{figure}
\begin{figure}
	\centering
	\includegraphics[width=0.7\textwidth , height=0.45\textwidth]{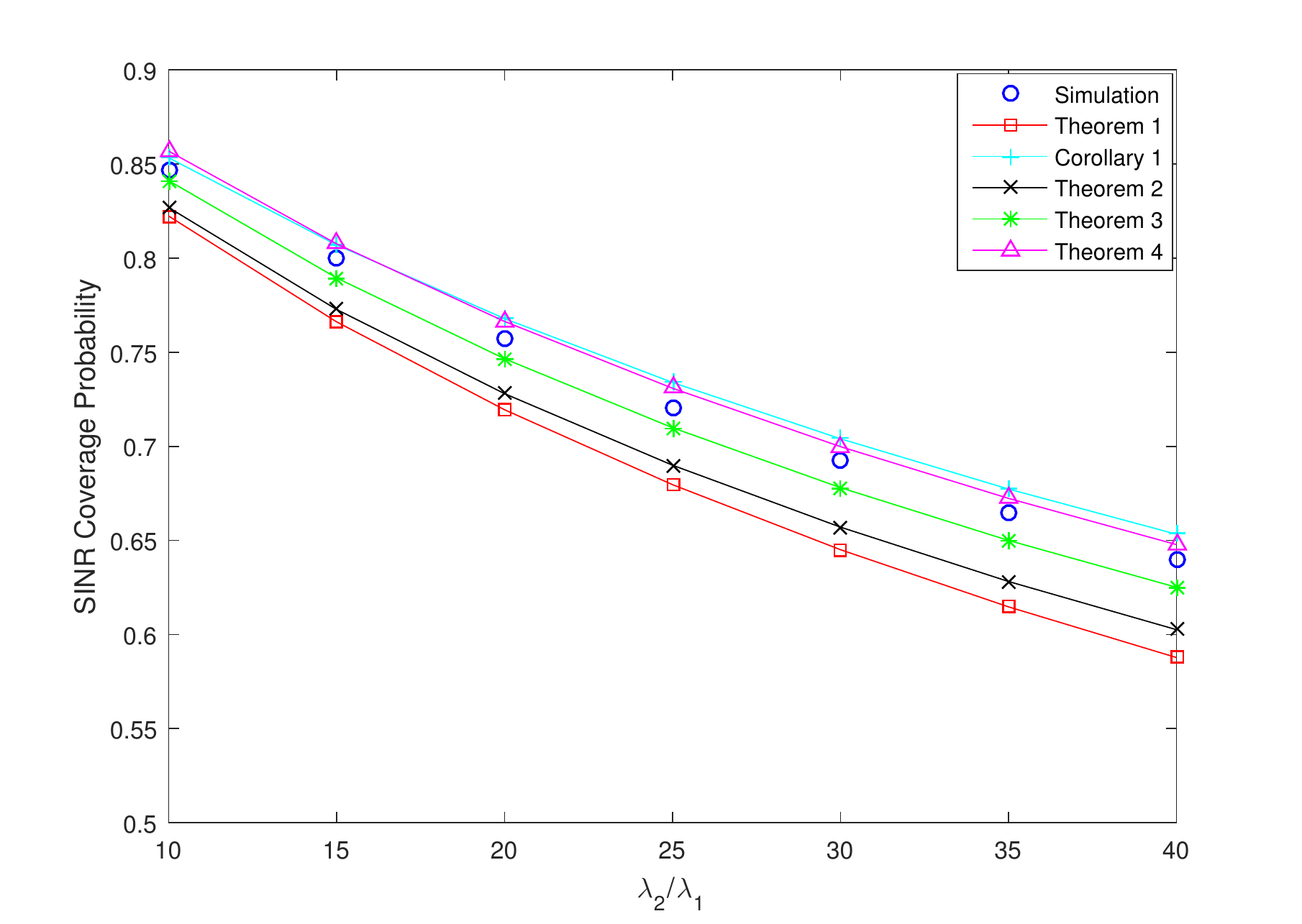}
	\caption{\small SINR coverage probability as a function of $ \lambda_{2}/\lambda_{1} $.}
	\label{lambda2lambda1_Fig}
\end{figure}
Since all results provide a remarkably accurate characterization of the coverage probability in Setup 1, we study only Setup 2 in the next figures, which is the most challenging configuration for evaluating the analytical expressions.
The coverage probability as a function of the central angle of holes have been plotted in Fig. \ref{theta_c_Fig}.
It is clear that the numerical result in Theorem \ref{Theorem1} cannot follow the simulation results and there is a substantial gap between PPP-based analysis and simulation. Comparison of the other proposed results with the simulations reveals that all of them have better performance than PPP-based analysis. Fig. \ref{lambda1_Fig} shows the effect of MBSs' density on the coverage probability. Small values of $ \lambda_{1} $ result in a low density of holes, whereas high values result in a high density of holes. Finally, we plot the coverage probability as a function of $ \lambda_{2}/\lambda_{1} $ in Fig. \ref{lambda2lambda1_Fig}. As was the case in the above results, all proposed analytical expressions work well with a minor error compared to the simulation result and provide better analysis than Theorem \ref{Theorem1}.

\section{Conclusion}\label{Conclusion}
In this paper, a novel PHP-based modeling of wireless networks was proposed. Contrary to the prior PHP models with circular shaped holes, we considered circular sector holes in a random direction. The relation between such hole configuration and circular hole explored and it revealed that the prior PHP models are a special case of ours and analysis based on our model provide a more general scheme than prior works. To capture spatial separation between tiers in mmWave HCN, we modeled SBSs and MBSs location as proposed PHP and PPP. Based on this model, an analytical framework was provided to compute the SINR coverage probability in the downlink of a mmWave two-tier HCN. Due to fundamental physical differences between mmWave and microwave propagation, we used directional beamforming, a blockage model and Nakagami fading to incorporate these differences. At first, fairly accurate analytical expressions for distance distribution of nearest LOS/NLOS BS to typical UE and association probability were derived and validated by simulation results. Then, SINR coverage probability was derived by utilizing some alternative approaches. The performance of these different approaches was compared with the simulation results and numerical results showed a dramatic improvement in the accuracy of our proposed approaches with that of the prior PPP-based analysis for mmWave HCN. Moreover, our analytical and simulation results revealed that NLOS holes have a more dominant effect on SINR distribution than LOS holes in a mmWave HCN. Providing an accurate analysis of the capacity-centric deployments in mmWave HCNs using other point processes, such as PCP and MCP, is a promising direction for future work.



%

\appendix
\subsection{Proof of Lemma \ref{lemma-A_{k}^{s}}} \label{proof of association probability}
Based on the considered cell association rule, the typical UE is associated with a $ s\in\{LOS,NLOS\} $ BS in $ k^{th} $ tier if the following is satisfied
\begin{equation}
\begin{split}
&P_{k}G_{k}r_{k}^{-\alpha^{s}}> P_{j}G_{j}r_{j}^{-\alpha^{s'}}
\overset{a}{=} P_{k}M_{k}M_{UE}r_{k}^{-\alpha^{s}}> P_{j}M_{j}M_{UE}r_{j}^{-\alpha^{s'}}\\
&=P_{k}M_{k}r_{k}^{-\alpha^{s}}> P_{j}M_{j}r_{j}^{-\alpha^{s'}}
= r_{j}>(\frac{P_{j}}{P_{k}}\times \frac{M_{j}}{M_{k}})^{1/\alpha^{s'}}\times r_{k}^{\alpha^{s}/\alpha^{s'}}
\end{split}
\end{equation}
where (a) follows by the serving link directivity gain assumption and $ j=1,2 $, $ s'\in\{LOS,NLOS\} $.

Let us denote the serving tier and LOS or NLOS state of BSs by T and S, respectively.
The typical UE is associated with a $ s\in\{LOS,NLOS\} $  BS in $ k^{th} $ tier if and only if it has a $ s\in\{LOS,NLOS\} $ BS in that tier, and its nearest BS in $ k^{th} $ tier has smaller average power than that of the nearest $ s'\neq s $ BS in $ k^{th} $ tier and the nearest $ s'\in\{LOS,NLOS\} $ BS in $ j\neq k $ tier. Hence, it follows that
\begin{equation}\label{Pr(T=k,S=s|r=x)}
\begin{split}
Pr(T=k,S=s|r_{k}=x)
=& Pr(r_{k}>x^{\alpha^{s}/\alpha^{s'}} , s' \neq s)
Pr(r_{j}>R_{j}^{s'}(x), j \neq k , s' \in{\{LOS,NLOS\}} )\\
=& exp(-\varLambda_{k,PHP}^{s'}([0,x^{\alpha^{s}/\alpha^{s'}}))) exp(- \sum_{s'\in\{LOS,NLOS\}} \varLambda^{s'}_{j,PHP}([0,R_{j}^{s'}(x) )))
\end{split}
\end{equation}
Note that we approximate PHP distribution of SBSs as PPP with $ \lambda_{PHP} $. Therefore, $ A_{k}^{s} $ can be expressed as
\begin{equation}
\begin{split}
A_{k}^{s}&= B_{k}^{s} Pr(T=k,S=s)
= B_{k}^{s} \mathbb{E}_{x}[Pr(T=k,S=s|r_{k}=x)]\\
&= B_{k}^{s} \int_{0}^{\infty} exp(-\varLambda_{k,PHP}^{s'}([0,x^{\alpha^{s}/\alpha^{s'}})))
exp(- \sum_{s'\in\{LOS,NLOS\}} 
\varLambda^{s'}_{j,PHP}([0,R_{j}^{s'}(x) )) )f_{k}^{s}(r_{k}=x)dx\\
&=\int_{0}^{\infty} \varLambda'^{s}_{k,PHP}([0,x))exp(-\sum_{j=1}^{2} \sum_{s'\in\{LOS,NLOS\}}\varLambda^{s'}_{j,PHP}([0,R_{j}^{s'}(x) )) ) dx
\end{split}
\end{equation}

\subsection{Proof of Theorem \ref{Theorem1}}\label{proof of SINR coverage probability-PPP}
SINR coverage probability is
\begin{equation}\label{SINR coverage probability }
P_{C}=Pr(SINR>\tau_{k}) \overset{a}{=} \sum_{k=1}^{2} \sum_{s\in\{LOS,NLOS\}} A_{k}^{s}P_{C_{k}}^{s}(\tau_{k})
\end{equation}
$ P_{C_{k}}^{s}(\tau_{k}) $ is calculated as follow
\begin{equation}\label{conditional coverage probability}
\begin{split}
P_{C_{k}}^{s}(\tau_{k})&=Pr(SINR>\tau_{k}\mid T=k , S=s)= \mathbb{E}_{x}[Pr(SINR>\tau_{k}\mid T=k , S=s,r_{k}=x)]\\
 &= \int_{0}^{\infty} Pr(SINR>\tau_{k}\mid T = k , S=s,r_{k}=x)\times f(r_{k}=x\mid T=k , S=s)dx
\end{split}
\end{equation}
and $ f(r_{k}=x\mid T=k , S=s) $ is
\begin{equation}\label{Conditional PDF}
\begin{split}
f(r_{k}=x\mid T=k , S=s)
&\overset{a}{=} \frac{Pr(T=k , S=s ,\mid r_{k}=x) f^{s}({r_{k}}=x)}{Pr(T=k , S=s)}\\
&\overset{b}{=} \frac{1}{A_{k}^{s}} \varLambda'^{s}_{k,PHP}([0,x))exp(-\sum_{j=1}^{2} \sum_{s'\in\{LOS,NLOS\}}\varLambda^{s'}_{j,PHP}([0,R_{j}^{s'}(x) ))
\end{split}
\end{equation}
where (a) follows from the Bayes theorem and (b) from equation (\ref{Pr(T=k,S=s|r=x)}) in Appendix \ref{proof of association probability} .
In order to complete our derivation for SINR coverage probability, we derive $ Pr(SINR>\tau_{k}\mid T = k , S=s,r_{k}=x) $. Using similar approach in \cite{6932503}, we have
\begin{equation}\label{kkkk}
\begin{split}
&Pr(SINR>\tau_{k}\mid T = k , S=s,r_{k}=x)= Pr(SINR_{k}^{s}>\tau_{k}\mid r_{k}=x)\\ &{\approx} \sum_{n=1}^{\upsilon^{s}} (-1)^{n+1}\binom{\upsilon^{s}}{n} \mathbb{E}[exp(-\mu_{k,n}^{s} \sigma^{2})]L_{I}(\mu_{k,n}^{s})
\end{split}
\end{equation}
where $ \mu_{k,n}^{s} $ was defined in Theorem \ref{Theorem1}. In order to calculate the Laplace transform of interferences $ L_{I}(\mu_{k,n}^{s})=\mathbb{E}[exp(-\mu_{k,n}^{s} I)]$, we define $ I $ as following
\begin{equation}\label{I}
\begin{split}
I=& \sum_{s'\in{\{LOS,NLOS\}}} \sum_{i\in{\{\phi_{1}^{s'}\} \backslash BS_{k,0}}} P_{1} |h_{1,i}|^{2} G_{1,i}r_{1,i}^{-\alpha^{s'}}
+\sum_{s'\in{\{LOS,NLOS\}}} \sum_{i\in{\{\psi^{s'}\} \backslash BS_{k,0}}} P_{2} |h_{2,i}|^{2} G_{2,i}r_{2,i}^{-\alpha^{s'}}\\
=& \sum_{j=1}^{2} \sum_{s'\in{\{LOS,NLOS\}}} \sum_{i\in{\{\phi_{j}^{s'}\} \backslash BS_{k,0}}} P_{j} |h_{j,i}|^{2} G_{j,i}r_{j,i}^{-\alpha^{s'}}
- \sum_{s'\in{\{LOS,NLOS\}}} \sum_{i\in{\{\phi_{2}^{s'}\cap \psi^{s'^{c}}\}}} P_{2} |h_{2,i}|^{2} G_{2,i}r_{2,i}^{-\alpha^{s'}}
\end{split}
\end{equation}
Let us denote
\begin{equation}\label{I_H}
 I_{H} = \sum_{s'\in{\{LOS,NLOS\}}} \sum_{i\in{\{\phi_{2}^{s'}\cap \psi^{s'^{c}}\}}} P_{2} |h_{2,i}|^{2} G_{2,i}r_{2,i}^{-\alpha^{s'}}
\end{equation}
which represents the interference from SBSs of the baseline PPP $ \phi_{2} $ who are inside of holes.

Since in this approach, PHP distribution of SBSs ($ \psi $) is approximated by the baseline PPP, $ \phi_{2} $, $ I_{H}=0 $. Hence, similar to PPP-based approaches like \cite{6932503}
\begin{equation}\label{E_I}
\begin{split}
&L_{I}(\mu_{k,n}^{s}) =\mathbb{E}[exp(-\mu_{k,n}^{s} (\sum_{j=1}^{2} \sum_{s'\in{\{LOS,NLOS\}}} \sum_{i\in{\{\phi_{j}^{s'}\}}}P_{j} |h_{j,i}|^{2} G_{j,i}r_{j,i}^{-\alpha^{s'}}))]\\
&= \prod_{j=1}^{2} \prod_{s'\in\{{LOS,NLOS}\}} exp(-\int_{R_{j}^{s'}(x)}^{\infty} (1- \mathbb{E}[exp(-\mu_{k,n}^{s}P_{j} |h_{j}|^{2} G_{j}r^{-\alpha^{s'}})] ) \varLambda'^{s'}_{j}([0,r)) dr)
\end{split}
\end{equation}
and $ \mathbb{E}[exp(-\mu_{k,n}^{s}P_{j} |h_{j}|^{2} G_{j}r^{-\alpha^{s'}})] $ is
\begin{equation}\label{h_G}
\begin{split}
&\mathbb{E}[exp(-\mu_{k,n}^{s}P_{j} |h_{j}|^{2} G_{j}r^{-\alpha^{s'}})]
\overset{a}{=} \sum_{g=1}^{4} P_{j,g}\mathbb{E}_{|h_{j}|^{2}}[exp(-\mu_{k,n}^{s}P_{j} |h_{j}|^{2} A_{j,g}r^{-\alpha^{s'}})]\\
&\overset{b}{=} \sum_{g=1}^{4} P_{j,g} (\frac{1}{1+(\mu_{k,n}^{s}P_{j} A_{j,g}r^{-\alpha^{s'}})/\upsilon^{s'}})^{\upsilon^{s'}}
\end{split}
\end{equation}
where in (a) expectation are taken over $ G_{j} $, $ P_{j,g} $ and $ A_{j,g} $ are constants defined in Table \ref{Probability mass function of directivity}, and step (b) follows from computing the moment generating function of the gamma distributed random variable $ |h_{j}|^{2} $.
The integration range excludes a ball centered at 0 and radius $R_{j}^{s'}(x)=(\frac{P_{j}}{P_{k}}\times \frac{M_{j}}{M_{k}})^{1/\alpha^{s'}}\times x^{\alpha^{s}/\alpha^{s'}} $ because the closest $ s'\in \{LOS,NLOS\} $ interferer in $ j^{th} $ tier has to be farther than the serving BS, based on cell association rule considered. Finally, by combining the above equations, SINR coverage probability expression given in Theorem \ref{Theorem1} is obtained.

\subsection{Proof of Theorem \ref{Theorem3}}\label{Proof of SINR Coverage Probability-PHP_serving hole}
In the case of incorporating the serving hole, we follow the same approach used in Appendix \ref{proof of SINR coverage probability-PPP} and equation (\ref{SINR coverage probability })-(\ref{kkkk}) are held in this proof too. It is enough to consider the effect of the serving hole in the interference characterization. Therefore, we approximate $ I $ as following, which incorporates the serving hole and ignores other holes
\begin{equation}
\begin{split}
I= \sum_{s'\in{\{LOS,NLOS\}}} \bigg\{\sum_{i\in{\{\phi_{1}^{s'}\} \backslash BS_{k,0}}} P_{1} |h_{1,i}|^{2} G_{1,i}r_{1,i}^{-\alpha^{s'}}+\sum_{i\in{\{\psi^{s'} \cap S^{c}(x,D,\theta_{c})\} \backslash BS_{k,0}}} P_{2} |h_{2,i}|^{2} G_{2,i}r_{2,i}^{-\alpha^{s'}}\bigg\}
\end{split}
\end{equation}
where $ S(x,D,\theta_{c}) $ was defined in Definition \ref{definition 1} and $ x $ is the distance between the serving MBS and the typical UE. Now we calculate $ L_{I}(\mu_{k,n}^{s}) $
\begin{equation}\label{I_serving hole}
\begin{split}
&L_{I}(\mu_{k,n}^{s})\\
&=\mathbb{E}[exp(-\mu_{k,n}^{s} ( \sum_{s'\in{\{LOS,NLOS\}}} \sum_{i\in{\{\phi_{1}^{s'}\}\backslash BS_{k,0}}} P_{1} |h_{1,i}|^{2} G_{1,i}r_{1,i}^{-\alpha^{s'}} + \\
& \sum_{s'\in{\{LOS,NLOS\}}} \sum_{i\in{\{\psi^{s'} \cap S^{c}(x,D,\theta_{c})\} \backslash BS_{k,0}}}P_{2} |h_{2,i}|^{2} G_{2,i}r_{2,i}^{-\alpha^{s'}}))]\\
&\overset{a}{=} \prod_{s'\in\{{LOS,NLOS}\}} \mathbb{E}[\prod_{i\in{\{\phi_{1}^{s'}\} \backslash BS_{k,0}}} exp(-\mu_{k,n}^{s}P_{1} |h_{1,i}|^{2} G_{1,i}r_{1,i}^{-\alpha^{s'}})]\\
& \prod_{s'\in\{{LOS,NLOS}\}} \mathbb{E}[\prod_{i\in{\{\psi^{s'} \cap S^{c}(x,D,\theta_{c})\} \backslash BS_{k,0}}} exp(-\mu_{k,n}^{s}P_{2} |h_{2,i}|^{2} G_{2,i}r_{2,i}^{-\alpha^{s'}})]\\
&\overset{b}{=} \prod_{s'\in\{{LOS,NLOS}\}} exp(-\int_{R_{1}^{s'}(x)}^{\infty} ( 1- \mathbb{E}[exp(-\mu_{k,n}^{s}P_{1} |h_{1}|^{2} G_{1}r^{-\alpha^{s'}})] ) \varLambda'^{s'}_{1}([0,r) dr) \\
& \prod_{s'\in\{{LOS,NLOS}\}} exp(-\int_{R_{2}^{s'}(x)}^{\infty}( 1- \mathbb{E}[exp(-\mu_{k,n}^{s}P_{2} |h_{2}|^{2} G_{2}r^{-\alpha^{s'}})] ) \varLambda'^{s'}_{2}([0,r) dr) \\
&\prod_{s'\in\{{LOS,NLOS}\}} exp(\lambda_{2} \int_{\Xi}(1- \mathbb{E}[exp(-\mu_{1,n}^{s}P_{2} |h_{2}|^{2} G_{2}r^{-\alpha^{s'}})]) P^{s'}(r) dS)\\
&=\prod_{j=1}^{2} \prod_{s'\in\{{LOS,NLOS}\}} exp(-\int_{R_{j}^{s'}(x)}^{\infty} ( 1- \mathbb{E}[exp(-\mu_{k,n}^{s}P_{j} |h_{j}|^{2} G_{j}r^{-\alpha^{s'}})] ) \varLambda'^{s'}_{j}([0,r) dr) \\
&\prod_{s'\in\{{LOS,NLOS}\}} exp(\lambda_{2} \int_{\Xi}(1- \mathbb{E}[exp(-\mu_{1,n}^{s}P_{2} |h_{2}|^{2} G_{2}r^{-\alpha^{s'}})]) P^{s'}(r) dS)
\end{split}
\end{equation}
where (a) is due to the independence assumption among tiers and LOS/NLOS BSs in each tier and (b) follows from the PGFL of a PPP. $ \Xi = S(x,D,\theta_{c})\bigcap B^{c}(0,R_{2}^{s'}(x)) $ and $ B^{c}(0,R_{2}^{s'}(x)) $ represents regions that are outside of a ball centered at the origin with radius $R_{2}^{s'}(x) = (\frac{P_{2}}{P_{k}}\times \frac{M_{2}}{M_{k}})^{1/\alpha^{s'}}\times x^{\alpha^{s}/\alpha^{s'}} $.

\begin{figure}
	\centering
	\includegraphics[width=0.6\textwidth]{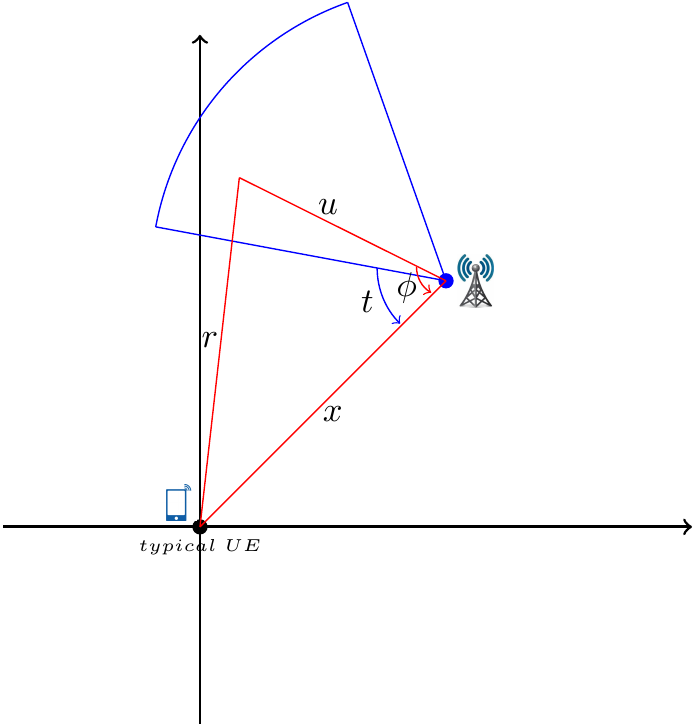}
	\caption{\small Illustration of the effect of a hole in the interference characterization.}
	\label{Coordinate_Systems}
\end{figure}

Next, we need to calculate $ \int_{\Xi} (1- \mathbb{E}[exp(-\mu_{1,n}^{s}P_{2} |h_{2}|^{2} G_{2}r^{-\alpha^{s'}})]) P^{s'}(r) dS $. For this, we use transformation as below
\begin{equation}
r=\sqrt{u^{2}+x^{2}-2uxcos(\phi)}
\end{equation}
where the above equation is derived based on cosine-law, $ u $ and $ \phi $ is defined in Fig. \ref{Coordinate_Systems}.

Following discussion in Remark \ref{Remark2}, we approximate $ \Xi \approx S(x,D,\theta_{c}) $. So, $ \int_{ \Xi}  (1- \mathbb{E}[exp(-\mu_{1,n}^{s}P_{2} |h_{2}|^{2} G_{2}r^{-\alpha^{s'}})]) P^{s'}(r) dS $ is
\begin{equation}
\begin{split}
&\int_{\Xi} (1- \mathbb{E}[exp(-\mu_{1,n}^{s}P_{2} |h_{2}|^{2} G_{2}r^{-\alpha^{s'}})]) P^{s'}(r) dS)\\
& \approx \int_{t}^{t+\theta_{c}} \int_{0}^{D} (1- \mathbb{E}[exp(-\mu_{1,n}^{s}P_{j} |h_{2}|^{2} G_{2}(u^{2}+x^{2}-2uxcos(\phi))^{-\alpha^{s'}/2})])\\
&P^{s'}(\sqrt{u^{2}+x^{2}-2uxcos(\phi)}) udud\phi)
\end{split}
\end{equation}
Since the direction of circular sectors have a uniform distribution in $ [0,2\pi) $, we have 
\begin{equation}
\begin{split}
&\int_{\Xi} (1- \mathbb{E}[exp(-\mu_{1,n}^{s}P_{2} |h_{2}|^{2} G_{2}r^{-\alpha^{s'}})]) P^{s'}(r) dS)\\
& \approx \int_{0}^{2\pi} \int_{t}^{t+\theta_{c}} \int_{0}^{D} (1- \mathbb{E}[exp(-\mu_{1,n}^{s}P_{j} |h_{2}|^{2} G_{2}(u^{2}+x^{2}-2uxcos(\phi))^{-\alpha^{s'}/2})])\\
&P^{s'}(\sqrt{u^{2}+x^{2}-2uxcos(\phi)}) \frac{1}{2\pi} udud\phi dt)\\
& \overset{a}{=} \frac{\theta_{c}}{2\pi} \int_{0}^{2\pi} \int_{0}^{D} (1- \mathbb{E}[exp(-\mu_{1,n}^{s}P_{j} |h_{2}|^{2} G_{2}(u^{2}+x^{2}-2uxcos(\phi))^{-\alpha^{s'}/2})])\\
&P^{s'}(\sqrt{u^{2}+x^{2}-2uxcos(\phi)}) udud\phi)
\end{split}
\end{equation}
where (a) is derived by substituting integrals and using this fact that integrand function is periodic respect to $ \phi $ with period $ 2\pi $. The above two-fold integral represents area of circular hole. Comparing results derived in \cite{7557010} for circular holes and using similar approach in equation (\ref{h_G}) to derive $ \mathbb{E}[exp(-\mu_{1,n}^{s}P_{2} |h_{2}|^{2} G_{2}(u^{2}+x^{2}-2uxcos(\phi))^{-\alpha^{s'}/2})] $, this integral can be expressed as follow
\begin{equation}
	\int_{x-D}^{x+D} F(\upsilon^{s'} , \frac{\mu_{k,n}^{s} P_{2}A_{2,g}u^{-\alpha^{s'}} }{\upsilon^{s'}}) 2\pi\lambda(u)
	P^{s'}(u) udu
\end{equation}
where $ \lambda(u) $ is defined in Theorem \ref{Theorem3}.
Therefore, SINR coverage probability by incorporating serving hole can be expressed as Theorem \ref{Theorem3}.

\subsection{Proof of Theorem \ref{Theorem4}}\label{Proof of SINR Coverage Probability-PHP_nearest_holes}
Similar to the proof used in Appendix \ref{Proof of SINR Coverage Probability-PHP_serving hole}, the approximation of the interference in this case is
\begin{equation}
I= \sum_{s'\in{\{LOS,NLOS\}}} \bigg\{\sum_{i\in{\{\phi_{1}^{s'}\} \backslash BS_{k,0}}} P_{1} |h_{1,i}|^{2} G_{1,i}r_{1,i}^{-\alpha^{s'}} + \sum_{i\in{\{\psi^{s'} \cap \Omega^{c}\} \backslash BS_{k,0}}} P_{2} |h_{2,i}|^{2} G_{2,i}r_{2,i}^{-\alpha^{s'}}\bigg\}
\end{equation}
where $ \Omega = \bigcup_{ s''\in{\{LOS,NLOS\}} } S(y,D,\theta_{c}) $ and $ y $ is the distance between $ s''\in{\{LOS,NLOS\}} $ interferer MBS and the typical UE. We ignore possible overlaps between two holes and approximate $ \Omega $ as $ \Omega \approx \sum_{ s''\in{\{LOS,NLOS\}} } S(y,D,\theta_{c}) $. Similar to equation (\ref{I_serving hole}), $ L_{I}(\mu_{k,n}^{s}) $ is
\begin{equation}
\begin{split}
&L_{I}(\mu_{k,n}^{s}|y) \\
&= \prod_{j=1}^{2} \prod_{s'\in\{{LOS,NLOS}\}} exp(-\int_{R_{j}^{s'}(x)}^{\infty} ( 1- \mathbb{E}[exp(-\mu_{k,n}^{s}P_{j} |h_{j}|^{2} G_{j}r^{-\alpha^{s'}})] ) \varLambda'^{s'}_{j}([0,r) dr) \\
&\prod_{s'\in\{{LOS,NLOS}\}} exp(\lambda_{2} \int_{\Omega}(1- \mathbb{E}[exp(-\mu_{k,n}^{s}P_{2} |h_{2}|^{2} G_{2}r^{-\alpha^{s'}})]) P^{s'}(r) dS)
\end{split}
\end{equation}

Note that the exact region for the integral is $ \Omega \bigcap B^{c}(0,R_{2}^{s'}(x)) $, but we approximate it by $ \Omega $.
Similar to proof in Appendix \ref{Proof of SINR Coverage Probability-PHP_serving hole},  $ \int_{\Omega}
(1- \mathbb{E}[exp(-\mu_{1,n}^{s}P_{2} |h_{2}|^{2} G_{2}r^{-\alpha^{s'}})]) P^{s'}(r) dS $ can be calculated as below
\begin{equation}
\begin{split}
&\int_{\Omega} (1- \mathbb{E}[exp(-\mu_{1,n}^{s}P_{2} |h_{2}|^{2} G_{2}r^{-\alpha^{s'}})]) P^{s'}(r) dS\\
& \approx \sum_{s''\in\{{LOS,NLOS}\}} \frac{\theta_{c}}{2\pi} \int_{0}^{2\pi} \int_{0}^{D} (1- \mathbb{E}[exp(-\mu_{k,n}^{s}P_{2} |h_{2}|^{2} G_{2}(u^{2}+y^{2}-2uycos(\phi))^{-\alpha^{s'}/2})])\\
&P^{s'}(\sqrt{u^{2}+y^{2}-2uycos(\phi)})udud\phi
\end{split}
\end{equation}
Now, to complete our derivation, we need to calculate PDF of $ y $. Given that the typical UE is associated with a $ s\in \{LOS,NLOS\} $ BS in $ k^{th} $ tier at distance $ r_{k}=x $ and observes at least one $ s''\in \{LOS,NLOS\} $ MBS, CCDF of y is
\begin{equation}
\begin{split}
&\bar{F}^{s''}(y|T=k,S=s,r_{k}=x)= Pr(Y>y|T=k,S=s,r_{k}=x) \\
&= Pr(\mathbb{N}(B(0,y) \backslash B(0,R_{1}^{s''}(x)) =\varnothing)=\frac{ exp(-\varLambda_{1}^{s''}([0,y)))exp(\varLambda_{1}^{s''}([0,R_{1}^{s''}(x))))}
{1-exp(-\varLambda_{1}([R_{1}^{s''}(x),\infty)))}
\end{split}
\end{equation}

where $ \mathbb{N} $ represents the number of points in $ \phi_{1} $ that are in the desired set. PDF of $ y $ now follows by differentiating the above expression
\begin{equation}
\begin{split}
&f^{s''}(y|T=k,S=s,r_{k}=x)= -\dfrac{d}{dy}\bar{F}^{s''}(y|T=k,S=s,r_{k}=x)\\ 
&=\frac{\varLambda'^{s''}_{1}([0,y)) exp(-\varLambda_{1}^{s''}([0,y))+\varLambda_{1}^{s''}([0,R_{1}^{s''}(x))))}
{1-exp(-\varLambda_{1}([R_{1}^{s''}(x),\infty)))}
\end{split}
\end{equation}

Finally $ L_{I}(\mu_{k,n}^{s}) $ can be expressed as
\begin{equation}
\begin{split}
&L_{I}(\mu_{k,n}^{s}) \\
&= \prod_{j=1}^{2} \prod_{s'\in\{{LOS,NLOS}\}} exp(-\int_{R_{j}^{s'}(x)}^{\infty}( 1- \mathbb{E}[exp(-\mu_{k,n}^{s}P_{j} |h_{j}|^{2} G_{j}r^{-\alpha^{s'}})] ) \varLambda'^{s'}_{j}([0,r)) dr)
\\
&\prod_{s'\in\{{LOS,NLOS}\}} \prod_{s''\in\{{LOS,NLOS}\}}\\
&\int_{R_{1}^{s''}(x)}^{\infty}exp(\frac{\theta_{c}}{2\pi}\lambda_{2} \int_{0}^{2\pi} \int_{0}^{D}(1- \mathbb{E}[exp(-\mu_{k,n}^{s}P_{2} |h_{2}|^{2} G_{2}(u^{2}+y^{2}-2uycos(\phi))^{-\alpha^{s'}/2})])\\
&P^{s'}(\sqrt{u^{2}+y^{2}-2uycos(\phi)})udud\phi)f^{s''}(y|T=k,S=s,r_{k}=x)dy
\end{split}
\end{equation}

Similar to equation (\ref{h_G}), $ \mathbb{E}[exp(-\mu_{k,n}^{s}P_{j} |h_{j}|^{2} G_{j}r^{-\alpha^{s'}})] $ and $ \mathbb{E}[exp(-\mu_{k,n}^{s}P_{2} |h_{2}|^{2} G_{2}(u^{2}+y^{2}-2uycos(\phi))^{-\alpha^{s'}/2})] $ can be calculated. Therefore, SINR coverage probability by incorporating the nearest non-serving LOS and NLOS can be obtained.

\subsection{Proof of Theorem \ref{Theorem5}}\label{Proof of SINR Coverage Probability-PHP_all_holes}
The exact expression for the interference in our proposed two-tier HCN model is
\begin{equation}
I= \sum_{s'\in{\{LOS,NLOS\}}} \bigg\{\sum_{i\in{\{\phi_{1}^{s'}\} \backslash BS_{k,0}}} P_{1} |h_{1,i}|^{2} G_{1,i}r_{1,i}^{-\alpha^{s'}} + \sum_{i\in{\{\phi_{2}^{s'} \cap \Gamma^{c}\} \backslash BS_{k,0}}} P_{2} |h_{2,i}|^{2} G_{2,i}r_{2,i}^{-\alpha^{s'}}\bigg\}
\end{equation}
where $ \Gamma = \bigcup_{y\in{\phi_{1}}} S(y,D,\theta_{c}) $. However, due to the possible overlaps between the holes, the exact characterization of $ I $ is complex. Therefore, we approximate $ \Gamma $ as $ \Gamma \approx \sum_{s''\in{\{LOS,NLOS\}}} \sum_{y\in{\phi_{1}^{s''}}} S(y,D,\theta_{c}) $, which ignores possible overlaps between the holes. Using this assumption, $ L_{I}(\mu_{k,n}^{s}) $ can be evaluated as 
\begin{equation}
\begin{split}
&L_{I}(\mu_{k,n}^{s}) \\
&\approx \prod_{j=1}^{2} \prod_{s'\in{\{LOS,NLOS\}}} exp(-\int_{R_{j}^{s'}(x)}^{\infty}( 1- \mathbb{E}[exp(-\mu_{k,n}^{s}P_{j} |h_{j}|^{2} G_{j}r^{-\alpha^{s'}})] ) \varLambda'^{s'}_{j}([0,r) dr) \\
& \prod_{s'\in\{{LOS,NLOS}\}} \mathbb{E}_{\phi_{1}}  [exp(\lambda_{2} \int_{\Gamma}(1- \mathbb{E}[exp(-\mu_{k,n}^{s}P_{2} |h_{2}|^{2} G_{2}r^{-\alpha^{s'}})]) P^{s'}(r) dS)]\\
& \approx exp (-\sum_{j=1}^{2} \sum_{s'\in\{LOS,NLOS\}} W_{j}^{s'}(x))\prod_{s'\in\{{LOS,NLOS}\}}\\
& \mathbb{E}_{\phi_{1}^{s''}} [ exp( \sum_{s''\in{\{LOS,NLOS\}}} \sum_{y\in{\phi_{1}^{s''}}}  \lambda_{2} \int_{S(y,D,\theta_{c})} (1- \mathbb{E}[exp(-\mu_{k,n}^{s}P_{2} |h_{2}|^{2} G_{2}r^{-\alpha^{s'}})]) P^{s'}(r) dS)]\\
& \approx exp (-\sum_{j=1}^{2} \sum_{s'\in\{LOS,NLOS\}} W_{j}^{s'}(x))\prod_{s'\in\{{LOS,NLOS}\}} \prod_{s''\in\{{LOS,NLOS}\}}\\
& \mathbb{E}_{\phi_{1}^{s''}} [\prod_{y\in{\phi_{1}^{s''}}} exp( \lambda_{2}\int_{S(y,D,\theta_{c})} (1- \mathbb{E}[exp(-\mu_{k,n}^{s}P_{2} |h_{2}|^{2} G_{2}r^{-\alpha^{s'}})]) P^{s'}(r) dS)]\\ 
& \approx exp (-\sum_{j=1}^{2} \sum_{s'\in\{LOS,NLOS\}} W_{j}^{s'}(x))\prod_{s'\in\{{LOS,NLOS}\}} \prod_{s''\in\{{LOS,NLOS}\}}\\
& exp(-\int_{R_{1}^{s''}(x)}^{\infty}  (1-  exp( \lambda_{2}\int_{S(y,D,\theta_{c})}
(1- \mathbb{E}[exp(-\mu_{k,n}^{s}P_{2} |h_{2}|^{2} G_{2}r^{-\alpha^{s'}})]) P^{s'}(r) dS)]))\varLambda_{1}^{s''}([0,y)) dy ) \\
& \approx exp (-\sum_{j=1}^{2} \sum_{s'\in\{LOS,NLOS\}} W_{j}^{s'}(x))\\
& exp(- \sum_{ s''\in{\{LOS,NLOS\}} } \int_{R_{1}^{s''}(x)}^{\infty}\sum_{ s'\in{\{LOS,NLOS\}} }  (1-exp(Q^{s'}(y)) )  \varLambda_{1}^{s''}([0,y)) dy )
\end{split}
\end{equation}

Note that similar to prior proofs, the exact region for the integral  is $ \Gamma \cap B^{c}(0,R_{2}^{s'}(x)) $, but we approximate it by $ \Gamma $.


\ifCLASSOPTIONcaptionsoff
  \newpage
\fi



%

\bibliographystyle{IEEEtran}
\bibliography{IEEEabrv}

\vfill
%








\end{document}